\documentclass[a4paper,11pt]{easychair}

\usepackage[makeroom]{cancel} 
\usepackage{hyperref}
\usepackage{amsmath}
\usepackage{amsthm}
\usepackage{amssymb}
\usepackage{mathrsfs}
\usepackage{graphicx}
\usepackage{stfloats}
\usepackage[normalem]{ulem} 
\usepackage[pdftex]{xcolor}
\usepackage{paralist}
\usepackage{enumerate}
\usepackage{url}
\usepackage{microtype}
\usepackage[all]{xy}
\def\romannum{\begingroup
  \def\theenumi{\textup{(\roman{enumi})}}%
  \def\p@enumi{}%
  \def\labelenumi{\theenumi}%
  \enumerate}

  {\end{enumerate}}%
\newenvironment{longromannum}[0]{
  \begin{enumerate}[(i)]
  }%
  {\end{enumerate}}

\newcommand{\str}[1]{{\mathcal{\uppercase{#1}}}}

\newcommand{\stabilo}{\bgroup\markoverwith
  {\textcolor{yellow}{\rule[-.5ex]{2pt}{2.5ex}}}\ULon}
\newcommand{\newstuff}{\bgroup\markoverwith{\textcolor{pink}{\rule[-.5ex]{2pt}{2.5ex}}}\ULon}

\newtheorem{theorem}{Theorem}
\newtheorem{definition}[theorem]{Definition}
\newtheorem{corollary}[theorem]{Corollary}
\newtheorem{application}[theorem]{Application}
\newtheorem{fact}[theorem]{Fact}
\newtheorem{lemma}[theorem]{Lemma}
\newtheorem{proposition}[theorem]{Proposition}

\theoremstyle{example}
\newtheorem{remark}[theorem]{Remark}
\newtheorem{example}[theorem]{Example}
\newtheorem{question}{Question}

\newcommand{\notmodels}{\ \makebox[0.1cm][l]{\ensuremath{\models}}/ \ }

\renewcommand{\phi}{\varphi}

\newcommand{\ignore}[1]{}

\newcommand{\algP}{\mathcal{P}}

\newcommand{\integerset}[0]{\ensuremath{\mathbb{N}}}
\newcommand{\QCSP}[0]{\ensuremath{\textrm{QCSP}}}
\newcommand{\CSP}[0]{\ensuremath{\textrm{CSP}}}
\newcommand{\restrict}{\ensuremath{\upharpoonright}}

\newcommand{\reactivelycomposable}{\ensuremath{\trianglelefteq}}

\title{From complexity to algebra and back: digraph classes, collapsibility and the PGP}
\titlerunning{From complexity to algebra and back}
\author{Catarina Carvalho\inst{1}
  \and
  Florent R. Madelaine\inst{2}\thanks{\textsl{thanks ANR grant ALCOCLAN}}
  \and
  Barnaby D. Martin\inst{3}\thanks{\textsl{thanks EPSRC grant EP/L005654/1}}
}

\institute{
  School of Physics, Astronomy and Mathematics
  University of Hertfordshire
  \email{c.carvalho2@herts.ac.uk}
  \and 
  GREYC, CNRS UMR 6072
  Universit\'e de Caen Basse Normandie
  \email{florent.madelaine@unicaen.fr}
  \and
  School of Science and Technology
  Middlesex University, London
  \email{B.Martin@mdx.ac.uk}
}
\authorrunning{Carvalho, Madelaine and Martin}

\begin{document}
\maketitle

\begin{abstract}
  Inspired by computational complexity results for the quantified
  constraint satisfaction problem, we study the clones of idempotent
  polymorphisms of certain digraph classes. Our first results are two
  algebraic dichotomy, even ``gap'', theorems. Building on and
  extending~\cite{QCSPforests}, we prove that partially reflexive
  paths bequeath a set of idempotent polymorphisms whose associated
  clone algebra has: either the polynomially generated powers property
  (PGP); or the exponentially generated powers property
  (EGP). 
  Similarly, we build on \cite{ICALP2014} to prove that semicomplete
  digraphs have the same property.

  These gap theorems are further motivated by new evidence that PGP
  could be the algebraic explanation that a QCSP is in NP even for
  unbounded alternation. 
  Along the way we effect also a study of a concrete form of PGP known as \emph{collapsibility},
  tying together the algebraic and structural threads from
  \cite{hubie-sicomp}, and show that collapsibility is equivalent to its $\Pi_2$-restriction. 
  We also give a decision procedure for $k$-collapsibility from a
  singleton source of a finite structure (a form of collapsibility
  which covers all known examples of PGP for finite structures). 

  Finally, we present a new QCSP trichotomy result, for partially
  reflexive paths with constants. Without constants it is known these
  QCSPs are either in NL or Pspace-complete \cite{QCSPforests}, but we
  prove that with constants they attain the three complexities NL,
  NP-complete and Pspace-complete.

\end{abstract}%

\section{Introduction}

A great literature of work exists from the past twenty years on
applications of universal algebra in the computational complexity of
\emph{constraint satisfaction problems} (CSPs) and a number of
celebrated results have been obtained through this method. 
Each
CSP is parameterised by a finite
structure $\mathcal{B}$ and asks whether an input
sentence $\varphi$ holds on $\mathcal{B}$, where $\varphi$ is a
primitive positive sentence, that is where only $\exists$ and $\land$ may be
used. For almost every class of model checking problem induced by the presence
or absence of first-order quantifiers and connectors, we can give a complexity
classification~\cite{DBLP:journals/corr/abs-1210-6893}: the two
outstanding classes are CSPs and its popular extension \textsl{quantified
  CSPs} (QCSPs) for positive Horn sentences -- where $\forall$ is also present -- which is used in
Artificial Intelligence to model non-monotone reasoning or
uncertainty. 

The outstanding conjecture in the area is that all finite-domain
CSPs are either in P or are NP-complete, something surprising given
these CSPs appear to form a large microcosm of NP, and NP itself is
unlikely to have this dichotomy property. This Feder-Vardi
conjecture \cite{FederVardi}, given more concretely in the algebraic
language in \cite{JBK}, remains unsettled, but is now known for
large classes of structures. 

The very useful role of algebra in unlocking the computational
complexities of QCSP has also been widely documented (see
\cite{rendezvous,Meditations}). Manuel Bodirsky has described the CSP
as a \emph{K\"onigsproblem} (king among problems) because it is an
important computational problem living at the interface of logic,
combinatorics and algebra. The QCSP is a somewhat less important
problem, with weaker links outside of the logical, where it is
formulated. In particular, its combinatorics are unwieldy -- for
example a totally satisfactory notion of a core remains elusive~\cite{DBLP:conf/cp/MadelaineM12} --
and its algebra is complicated by the fact that the class of
surjective polymorphisms is not closed under composition. This
perhaps explains why the complexity of QCSPs is classified for rather modest classes
of structures, for which only three complexities are observed P, NP-complete
and Pspace-complete.
 
In the case in which only \emph{idempotent polymorphisms} are
considered --
corresponding relationally to all constants being definable in
$\mathcal{B}$ -- some
better behaviour is restored and it is mostly in this arena that we shall
place ourselves.
What seems to be a unifying explanation for a
complexity in NP is that it suffices to check an instance $\varphi$
with $m$ universal variables for a small fraction (polynomial in $m$
and $\mathcal{B}$) of all possible choices for these $m$ universal
variables. This property can be viewed as a special form of quantifier
relativisation in the sense that it suffices to check an instance against
restricted Skolem functions. This fits in well with the classification for
model checking for other fragments of FO  where relativisation also
characterises the complexity~\cite{DBLP:journals/corr/abs-1210-6893}.

In Hubie Chen's \cite{AU-Chen-PGP}, a new traverse between algebra and
QCSP was discovered. Chen's previous work in QCSP tractability largely
involved the special notion of \emph{collapsibility}
\cite{hubie-sicomp}, but in \cite{AU-Chen-PGP} this was extended to a
version of the \emph{polynomially generated powers} (PGP)
property. This latter ties in with a rich literature of dichotomy
(``gap'') theorems on growth rate of generating sets of direct powers of algebras. The PGP properly generalises collapsibility and reveals a link to universal algebra that we explore in this paper and we might argue makes QCSP at least a \emph{F\"urstenproblem}  (prince among problems).

The initial algebraic phenomenon of our study is the growth rate of
generating sets for direct powers of an algebra. That is, for an
algebra $\mathbb{A}$ we associate a function
$f_\mathbb{A}:\mathbb{N}\rightarrow\mathbb{N}$, giving the cardinality
of the minimal generating sets of the sequence $\mathbb{A},
\mathbb{A}^2, \mathbb{A}^3, \ldots$ as $f(1), f(2), f(3), \ldots$,
respectively. We may say $\mathbb{A}$ has the \emph{$g$-generating property}
($g$-GP for short) if $f(m) \leq g(m)$ for all $m$. The question then
arises as to the growth rate of $f$ and specifically regarding the
behaviours constant, logarithmic, linear, polynomial and
exponential. Wiegold proved in \cite{WiegoldSemigroups} that if
$\mathbb{A}$ is a finite semigroup then $f_{\mathbb{A}}$ is either
linear or exponential, with the former prevailing precisely when
$\mathbb{A}$ is a monoid. This dichotomy classification may be seen as
a gap theorem because no growth rates intermediate between linear and
exponential may occur. We say $\mathbb{A}$  enjoys the
\emph{polynomially generated powers} property 
(PGP) if there exists a polynomial $p$ so that $f_{\mathbb{A}}=O(p)$
and  the \emph{exponentially generated powers} property (EGP) if there
exists a constant $b$ so that $f_{\mathbb{A}}=\Omega(g)$ where
$g(i)=b^i$.

The PGP implies that the bounded alternation QCSP is in NP
rather than the corresponding level of the polynomial hierarchy one
expects in general, provided that generators may be generated
effectively, \emph{effective PGP} in Chen's parlance.
This should be clear for $\Pi_2$-sentences
(quantifier prefix of the form $\forall^\star \exists^\star$) as it
suffices to solve one CSP per generator, and by induction this holds
for bounded alternation.
Moreover, for all known examples it also holds for unbounded
alternation. In particular, for known examples of finite structures, this drop is witnessed by an operation which characterises a type of collapsibility (from the so-called \emph{singleton source}), which we shall call
a \emph{Hubie operation}. When this is present as a polymorphism, it
implies a drop to NP also in the unbounded case as it may be composed
in a more involved fashion suitable for working with Skolem functions,
what Chen terms \emph{reactive composition}.

Hubie Chen proved the first PGP-EGP gap theorem for polymorphism
clones in  \cite{AU-Chen-PGP}. Namely, let id-Pol$(\mathcal{B})$\footnote{We will view this as at once a set of polymorphisms on domain $B$ and an algebra of operations over that domain.} be
the clone of idempotent polymorphisms of a $3$-element structure
$\mathcal{B}$ such that id-Pol$(\mathcal{B})$ does not contain a G-set
as a factor\footnotemark{}\footnotetext{This is a technical
  assumption that we will not define. When there is a G-set as a
  factor we know the corresponding QCSP is NP-hard \cite{JBK}.}. Then either
id-Pol$(\mathcal{B})$ has PGP or it has EGP. Indeed, this result
extended the previous observation of Chen that when
id-Pol$(\mathcal{B})$ is the clone of idempotent polymorphisms of a
$2$-element structure $\mathcal{B}$, then either id-Pol$(\mathcal{B})$
has  PGP or it has EGP. Now, \emph{$k$-$\Pi_2$-collapsibility} (whose naming will be explained in the sequel) can be seen
as a special form of the PGP in which the generating set for each
$\mathbb{A}^m$ may be taken to be the set of $m$-tuples  which contain
the repetition of a single element from a so-called source set at least $m-k$
times, the other at most $k$ positions being arbitrary. $k$-collapsibility can be seen similarly but manifests slightly differently through the already alluded to \emph{reactive composition}
of this set of $m$-tuples. In the $2$-element case, the
PGP manifests in the special form of $1$-collapsibility,  but already
in the $3$-element case there are algebras with the PGP that are not
$k$-collapsible for any $k$, though no such example is known for a
finite structure (\textsl{i.e.} with finitely many relations). 

When a structure $\mathcal{H}$ expanded by all constants is so that
QCSP$(\mathcal{H})$ is Pspace-complete, then (under the
complexity-theoretic assumption that NP is different from Pspace) we can assume that
id-Pol$(\mathcal{H})$ does not have \emph{effective
  PGP} \cite{Meditations}. 
Naturally, these are the places to look to prove
EGP results.  The QCSP complexity classification for $3$-element
structures is still open, even in the idempotent case, but this paper
builds upon Chen's \cite{AU-Chen-PGP} motivated by the extant complexity
classifications for the QCSP for partially reflexive trees in \cite{QCSPforests} and
semicomplete digraphs in \cite{ICALP2014}. Thus, the complexity results lead the algebra, in contrast to the typical modus operandi.

\section*{Principal contributions}
\label{sec:our-contributions}

\subsection*{Complexity to algebra: new PGP-EGP gaps.}
For partially reflexive paths we recall the notion of being
quasi-loop-connected from~\cite{QCSPforests}, and prove the following
algebraic gap.
\begin{theorem}
  Let $\mathcal{H}$ be a partially reflexive path. If $\mathcal{H}$ is
  quasi-loop-connected, then  $\mathrm{id\mbox{-}Pol}(\mathcal{H})$ has the
  PGP. Otherwise,  $\mathrm{id\mbox{-}Pol}(\mathcal{H})$ has the EGP. 
  \label{thm:prp-hauptsatz}
\end{theorem}
\noindent Along the way, we
also characterise precisely which partially reflexive paths have only essentially unary polymorphisms.

Building upon and refining \cite{ICALP2014}, we derive a second gap for
semicomplete digraphs.
\begin{theorem}
  \label{thm:sc-hauptsatz}
  Let $\mathcal{H}$ be a semicomplete digraph. If $\mathcal{H}$ has at
  most one cycle  or both a source and a sink, then
  $\mathrm{id\mbox{-}Pol}(\mathcal{H})$ has the PGP. Otherwise, $\mathrm{id\mbox{-}Pol}(\mathcal{H})$
  has the EGP. 
\end{theorem}

\newcommand{\BarnyApprovedTitleForOurStuffIntro}[0]{The PGP: collapsibility and beyond}
\newcommand{\BarnyApprovedTitleForOurStuffText}[0]{\BarnyApprovedTitleForOurStuffIntro}
\subsection*{\BarnyApprovedTitleForOurStuffIntro.}
\begin{table*}[bp]
  \centering
  \begin{tabular}[h]{|p{.55\textwidth}cp{.305\textwidth}|}
    \hline
    Polymorphism & Arity & Collapsibility\\
    \hline
    \hline
    Near unanimity (a.k.a. majority when $k=3$)
    & 
    $k$ 
    & 
    $(k-1)$-collapsibility with source $\{x\}$ for any $x$.\\
    \multicolumn{3}{|p{\textwidth}|}{
    satisfies the identities $f(x,y,\ldots,y)= \ldots
      =f(\ldots,y,x,y\ldots) = f(y,\ldots,y,x) =y$}
    \\
    \hline
    Dual discriminator.
    & 
    3 
    & 
    $1$-collapsibility with source $A$. \\
     majority acting as a projection when the 3 arguments are distinct &  & $2$-collapsibility with source $\{x\}$ for any $x$.\\
    \hline
    Mal'tsev  
    & 
    3 
    & $1$-collapsibility with source $\{x\}$ for any $x$.\\
    $m(x,x,y)=m(y,x,x)=y$ &&\\
    \hline
    Hubie operation : remains surjective when any coordinate is fixed
    to be $x$ 
    &
    $k$ 
    & 
    $(k-1)$-collapsibility with source $\{x\}$.\\
    In particular, the case of so-called \emph{semilattice with unit}
    $\{x\}$: a binary idempotent, associative and commutative
    polymorphism $s$ that satisfies $s(x,y) = s(y,x) = y$ for any $y$.
    &
    2
    &
    1-collapsibility with source $\{x\}$.
    \\
    \hline
  \end{tabular}
  \medskip
  \caption{Some polymorphisms that imply collapsibility.%
  }
  \label{tab:Polymorphismscollapsibility}
\end{table*}
We prove that when we have a sufficiently uniform form
of effective PGP, based on the notion of \emph{projective} sequences
of adversaries (an adversary is a set of tuples restricting the tuple
of universal variables), then we also have a drop in complexity to NP even in
the unbounded case. 
For such sequences of adversaries, we can show that they
are generating iff they are generating via reactive composition. 
Our proof relies on and adapts the notion of a canonical
$\Pi_2$-sentence from~\cite{LICS2008}.
The statement of this result, Theorem~\ref{MainResult:InAbstracto}, is
somewhat technical so we state here its concrete application to the
situation of collapsibility.

\vspace{0.2cm}
\noindent \textbf{Corollary~\ref{MainResult:InConcreto:Collapsibility}} (Part of).
  Let $\mathcal{A}$ be a structure, $\emptyset \subsetneq B\subseteq A$ and $p>0$.
  The following are equivalent.
  \begin{romannum}
  \item $\mathcal{A}$ is $p$-collapsible from source $B$.
  \item $\mathcal{A}$ is $\Pi_2$-$p$-collapsible from source $B$.
  \item For every $m$, the structure $\mathcal{A}$ satisfies a
    canonical $\Pi_2$-sentence with $m\cdot |A|$ universal variables.
  \end{romannum}

In the case of a singleton source, which covers all known examples of
collapsibility for finite structures (see also
Table~\ref{tab:Polymorphismscollapsibility} which recalls
the polymorphisms that are known to imply collapsibility), then we can refine this further as follows.

\

\noindent \textbf{Theorem~\ref{theo:singletonSource:StrongCharacterisation}} (Part of). \textbf{($p$-Collapsibility from a singleton source}).
Let $p\geq 1$ and $x$ be a constant in $\mathcal{A}$.
  The following are equivalent:
  \begin{longromannum}
  \item $\mathcal{A}$ is $p$-collapsible from $\{x\}$.
  \item $\mathcal{A}$ is $\Pi_2$-$p$-collapsible from $\{x\}$. 
  \item $\mathcal{A}$ models a \textbf{single} canonical $\Pi_2$-sentence
    which implies that $\mathcal{A}$ admits a Hubie operation as a
    polymorphism. 
  \end{longromannum}

This means that we may decide $p$-collapsibility from a singleton
source (the parameter $p>0$ being part of the input).

\subsection*{Back to complexity.}
As we have argued already, a uniform form of PGP like
$p$-collapsibility might explain when a QCSP is in NP. It is
natural in this context to allow constants in the structure not only
because it makes things well behaved in the algebra, but also because
constants are needed for the natural algorithm which consists in
solving a polynomial number of CSP instances induced by replacing all
but $p$ variables by a constant.
Finally, we apply our earlier results: that collapsibility coincides with its $\Pi_2$-restriction and that partially
reflexive paths that are not quasi-loop-connected remain $\Pi_2$-collapsible
in the idempotent case. This morphs the first dichotomy theorem of
\cite{QCSPforests} (\mbox{cf.}
Theorem~\ref{theo:QCSPforestsDichotomyNoConstants}.) 
to become a new trichotomy theorem. Specifically,
the NL cases in the absence of constants split to become NL and
NP-complete cases in the presence of constants. 
\begin{theorem}
  \label{thm:complexity}
  Let $\mathcal{H}$ be a partially reflexive path expanded with all constants. 
  \begin{romannum}
  \item If $\mathcal{H}$ is loop-connected, then QCSP$(\mathcal{H})$
    is in NL.
    \label{thm:complexity:i}
  \item Else, if $\mathcal{H}$ is quasi-loop-connected, then
    QCSP$(\mathcal{H})$ is NP-complete.
    \label{thm:complexity:ii}
  \item Otherwise,  QCSP$(\mathcal{H})$ is Pspace-complete.
    \label{thm:complexity:iii}
  \end{romannum}
\end{theorem}

Due to space restriction, many proofs have been omitted and can be
found in the appendix.
\section{Preliminaries}
Throughout we consider only finite relational structures possibly with
some constants. On first reading, the reader might prefer to assume that
  all constants are present, for the sake of simplicity; though we can not make this
  assumption in general as adding all constants may increase the complexity
  (compare Theorem~\ref{thm:complexity} with Theorem~\ref{theo:QCSPforestsDichotomyNoConstants}).
We denote by $\sigma$ our base signature and hereafter
unless otherwise specified, a structure will be a
$\sigma$-structure. We shall denote by $A$ the domain of a structure
$\mathcal{A}$. The \emph{canonical query}\footnotemark{}
\footnotetext{We actually consider the quantifier-free part of the
  \emph{canonical query}. We depart from the
  usual definition where an existential sentence is used, as we will often need a different prefix of
  quantification.} of the structure
$\mathcal{A}$ is the quantifier-free first-order sentence that has one
variable $x_a$ for each element $a$ in $A$ and a conjunction of all
the positive facts of $\mathcal{A}$: \textsl{e.g.} 
$R(a_1,a_2,\ldots,a_r)$ holds in $\mathcal{A}$ for some $r$-ary symbol
in $\sigma$ iff this conjunction contains the conjunct
$R(x_{a_1},x_{a_2},\ldots,x_{a_r})$. Conversely, given a conjunction
of positive atoms $\varphi$, we denote by $\mathcal{D}_\varphi$ its \emph{canonical
  database}, that is the structure with domain the variables of
$\varphi$ and whose tuples are precisely those that are atoms of
$\varphi$.
%
Let $\mathcal{A}$ and $\mathcal{B}$ be structures. A
\emph{homomorphism} $h$ from $\mathcal{A}$ to $\mathcal{B}$ is a map
from $A$ to $B$ such that for every relational symbol $R$ of arity $r$
and every $r$-tuple $(a_1,a_2,\ldots,a_r)$ of elements of $A$ such
that $R(a_1,a_2,\ldots,a_r)$ holds in $\mathcal{A}$ we have that
$R(h(a_1),h(a_2),\ldots,h(a_r))$ holds in $\mathcal{B}$. 
The \emph{product} $\mathcal{A}\otimes \mathcal{B}$ is the structure
with domain $A\times B$ such that for every relational symbol $R$ of
arity $r$ and every $r$-tuples $(a_1,a_2,\ldots,a_r)$ of  elements of
$A$ and $(b_1,b_2,\ldots,b_r)$ of elements of $B$, we have that
$R\bigl((a_1,b_1),(a_2,b_2),\ldots,(a_r,b_r)\bigr)$ holds in
$\mathcal{A}\otimes\mathcal{B}$ iff both 
$R(a_1,a_2,\ldots,a_r)$ holds in $\mathcal{A}$ and
$R(b_1,b_2,\ldots,b_r)$ holds in $\mathcal{B}$. A constant symbol $c$
is interpreted in $\mathcal{A}\otimes\mathcal{B}$ as the element
$(a,b)$  where $a$ and $b$ are the interpretation of $c$ in
$\mathcal{A}$ and $\mathcal{B}$, respectively.  
We write $\mathcal{A}^k$ for the product of $k$ copies of
$\mathcal{A}$.
A \emph{$k$-ary polymorphism} of $\mathcal{A}$ is a homomorphism $f$
from $\mathcal{A}^k$ to $\mathcal{A}$. We say that $f$ is
\emph{idempotent} if for any $x$ in $A$, $f(x,x,\ldots,x)=x$ holds. 
Equivalently, $f$ is a polymorphism of an extension of $\mathcal{A}$
with constants symbols naming the elements of $\mathcal{A}$.
Let $\mathrm{id\mbox{-}Pol}(\mathcal{A})$
(resp. $\mathrm{sPol}(\mathcal{A})$) denote the set of idempotent
(resp. surjective) polymorphisms of $\mathcal{A}$. 
%
A \emph{majority} operation is a ternary operation $f$ that satisfies the identities $f(x,x,y)=f(x,y,x)=f(y,x,x)=f(x,x,x)=x$. The \emph{dual discriminator} ($dd$) is the particular majority that satisfies $dd(x,y,z)=x$ when $x,y,z$ are distinct.
%
A \emph{Hubie operation} is a
surjective $k$-ary operation $f$, on a set $A \ni x$, such that
$f(x,x,\ldots,x)=x$ and
$f(x,A,\ldots,A)=f(A,x,\ldots,A)=\ldots=f(A,A,\ldots,x)=A$. That is,
the restriction of the operation from fixing $x$ in each coordinate
position remains surjective. When we need to specify $x$, we speak of
a \emph{Hubie operation with source $x$}.
%
A \emph{positive Horn sentence} (pH-sentence for short) is a sentence
of first-order logic with equality using both quantifiers $\exists$ and $\forall$
but only the logical connective $\land$. We will only  consider pH
sentences in \emph{prenex form}, that is with all quantifiers in
front. In the absence of the universal quantifier, we speak of a
\emph{primitive positive sentence} (pp-sentence for short). 
A $\Pi_2$-pH sentence is a pH-sentence with quantifier prefix of the
form $\forall^\star\exists^\star$, that is a block of universal
variables followed by a block of existential variables. 
Let $\mathcal{A}$ be a finite relational structure (possibly with
constants).
The \emph{quantified constraint satisfaction problem} with structure $\mathcal{A}$, denoted $\QCSP(\mathcal{A})$, is the
model-checking problem for pH-sentences over $\mathcal{A}$. That  is,
it takes as input a pH-sentence $\varphi$ and asks whether
$\mathcal{A}$ models $\varphi$. When $\mathcal{A}$ is a structure
with constants naming its elements, we may write
$\QCSP_c(\mathcal{A})$ to stress that all constants are present. 
Similarly, let $\CSP(\mathcal{A})$ denote
the \emph{constraint satisfaction problem} with structure
$\mathcal{A}$ defined as above but with pp-sentences.
We will denote by $\langle \mathcal{A} \rangle_{\mathrm{pH}}$ the
class of relations that are interpretable in $\mathcal{A}$ via some pH-sentence.

Reading the introduction, one could be forgiven for thinking
collapsibility is at once a logical property of structures and a
property of algebras. Indeed, Chen \cite{hubie-sicomp} defines a form
of collapsibility for each and shows that the algebraic form implies
the logical one (a result reworded here as Theorem~\ref{thm:hubieReactiveComposition}).
One purpose of this paper is to tie these two definitions together and prove the converse. For formal purposes we will define collapsibility only in the
logical sense. Let  $\mathcal{A}$ be a structure, $B\subseteq A$ and
$p \geq 0$. The structure $\mathcal{A}$ is \emph{$p$-collapsible with source $B$} when for all
$m\geq 1$, for all pH-sentences $\phi$ with $m$ universal quantifiers, we have that $\mathcal{A} \models \phi$ iff $\mathcal{A} \models \psi$, for \textbf{all} sentences $\psi$ obtained by instantiating all but $p$ universal variables of $\phi$ by some single element $x \in B$.
We assume here that $\mathcal{A}$ has all constants from
  the source set $B$ and will delay to
  \S~\ref{sec:games-adversaries-reactivecomposition} for a more
  general definition where this assumption is not necessary.
$\mathcal{A}$ is \emph{collapsible with source $B$} if it is
$p$-collapsible with source $B$ for some $p$. We define similarly the analogous
notions for the $\Pi_2$-fragment. 


\section{New PGP-EGP gaps}
\label{sec:pgp-versus-egp}
Let $[n]:=\{1,\ldots,n\}$. A digraph $\mathcal{G}$ has vertex set $G$, of cardinality $|G|$, and edge set $E(\mathcal{G})$. Similarly, an algebra $\mathbb{A}$ has domain $A$. For a digraph $\mathcal{H}$, the distance between two $m$-tuples $\overline{s}=(s^1,\ldots,s^m)$ and $\overline{t}=(t^1,\ldots,t^m) \in H^m$ is the minimal $r$ so that there are $m$-tuples  $\overline{z}_1=(z_1^1,\ldots,z_1^m),\ldots,\overline{z}_{r-1}=(z_{r-1}^1,\ldots,z_{r-1}^m) \in H^m$ such that, for each $i \in [m], j \in [r-2]$, we have $E(s^i,z_1^i)$, $E(z_j^i,z_{j+1}^i)$ and $E(z_{r-1}^i,t^i)$.

\subsection{Partially reflexive paths}
\label{sec:prpaths}
Henceforth we consider partially reflexive paths, \mbox{i.e.} paths potentially with some loops (we will frequently drop the preface partially reflexive). As we are interested in idempotent polymorphisms these paths come with constants naming each of their vertices.
For a sequence $\beta \in \{0,1\}^*$, of length $|\beta|$, let $\mathcal{P}_\beta$ be the undirected path on $|\beta|$ vertices such that the $i^{th}$ vertex has a loop iff the $i^{th}$ entry of $\beta$ is $1$ (we may say that the path $\mathcal{P}$ is \emph{of the form} $\beta$). A path $\mathcal{H}$ is \emph{quasi-loop-connected} if it is of either of the forms 
\begin{romannum}
\item $0^a 1^b \alpha$, for $b>0$ and some $\alpha$ with $|\alpha|=a$,
  or
  \label{def:qlc:i}
\item $0^a \alpha$, for some  $\alpha$ with $|\alpha|\in \{a,a-1\}$.
  \label{def:qlc:ii}
\end{romannum}
Where a path satisfies both \ref{def:qlc:i} and \ref{def:qlc:ii}, we use formulation \ref{def:qlc:i} preferentially. A path whose self-loops induce a connected component is further said to be \emph{loop-connected}.
We will usually envisage the domain of a path with $n$
vertices to be $[n]$, where the vertices appear in the natural order
(and a good behaviour brought by the absence of self-loops of the quasi-loop connected case is exhibited
in the lower numbers). 
The \emph{centre} of a path is either the middle vertex, if there is
an odd number of vertices, or between the two middle vertices,
otherwise.  
The main result of this section was stated as Theorem~\ref{thm:prp-hauptsatz}.
\begin{proof}[Proof of Theorem~\ref{thm:prp-hauptsatz}]
The PGP cases follow from Lemmas~\ref{lem:prp-1}, \ref{lem:generation}
and \ref{lem:generation2}. The EGP cases follow from Proposition~\ref{prop:pathsEGP}.
\end{proof}

\subsubsection{Partially reflexive paths with the PGP}
\label{sec:prPAthsWithPGP}

\

\noindent{}The loop-connected case is well understood.
\begin{lemma}
\label{lem:prp-1}
Let $\mathcal{H}$ be a partially reflexive  path that is loop-connected. Then $\mathrm{id\mbox{-}Pol}(\mathcal{H})$ has the PGP.
\end{lemma}
\begin{proof}
  $\mathcal{H}$ admits a majority polymorphism (see Lemma 3 of
  \cite{QCSPforests}). This is a Hubie polymorphism of $\mathcal{G}$ (where the single element can be chosen arbitrarily), whereupon the result follows from \cite{hubie-sicomp} (see our forthcoming Lemma~\ref{lem:ChensLemma} together with Corollary~\ref{MainResult:InConcreto:Collapsibility}).
\end{proof}
The quasi-loop connected case is more technical. Due to space
restriction, we will only present in full half of this case, which
will suffice to illustrate the proof principle. First, we are able to
exhibit specific binary idempotent polymorphisms.
\begin{lemma}
\label{lem:bin-pol}
Let $\mathcal{P}_{0^a 1^b \alpha}$, with $b>0$, be a quasi-loop-connected path on vertices $[n]$. For each $y \in [n]$ there is a binary idempotent polymorphism $f_y$ of $\mathcal{P}_{0^a 1^b \alpha}$ so that $f_y(1,x)=x$ (for all $x$) and $f_y(n,1)=y$.
\end{lemma}
Next, we exhibit specific linear generating set for the powers.
\begin{lemma}
  \label{lem:generation}
  Let $\mathcal{P}_{0^a 1^b \alpha}$, for $b>0$, be a quasi-loop-connected path on vertices $[n]$. Let $\mathbb{A}$ be the algebra specified by $\mathrm{id\mbox{-}Pol}(\mathcal{P}_{0^a 1^b \alpha})$. For each $m$, $\mathbb{A}^m$ is generated from the $n+1$ $m$-tuples $(1,1,\ldots,1),$ $(n,1,\ldots,1), (1,n,\ldots,1), \ldots, (1,1,\ldots,n)$.
\end{lemma}
\begin{proof}
We will make use of the polymorphisms $f_y$ guaranteed to exist by Lemma~\ref{lem:bin-pol}. Firstly, from $(n,1,\ldots,1)$ and $(1,1,\ldots,1)$ we can, for each $y$, use $f_y$ to generate $(y,1,\ldots,1)$. And we can similarly build all co-ordinate permutations of this. We now have the base case in an inductive proof, where our inductive hypothesis will be that for all $k$ we can build the tuple which has entries $y_1,\ldots,y_k$ with the remaining entries being $1$. The result for $k=m$ implies the lemma, so it remains only to test the inductive step where we will assume $y_1,\ldots,y_k,y_{k+1}$ are the first $k+1$ entries of a tuple continued by $1,\ldots,1$ (of course we can build the rest through co-ordinate permutation). From $(1,\ldots,1,n,1,\ldots,1)$ and $(y_1,\ldots,y_k,1,\ldots,1)$ (where $n$ is in the $k+1$st position) we can use $f_{y_{k+1}}$ to build  $(y_1,\ldots,y_k,y_{k+1},1,\ldots,1)$. This proves the claim.
\end{proof}

Lemma~\ref{lem:bin-pol} fails for the other type of
quasi-loop-connected paths, essentially when $b=0$. This is easily
seen to be the case when we take an irreflexive path on an odd number
$n$ of vertices (for an example on paths with an even number of
vertices $\geq 4$, take an irreflexive path leading to a single looped
vertex at the end). Then no idempotent polymorphism $f$ may have
$f(n,1)=2$ for parity reasons, since odd and even vertices must be at
odd distance in the square of the graph. In fact,
Lemma~\ref{lem:bin-pol} does hold for quite a few of the remaining
cases (\mbox{e.g.} for $\mathcal{P}_{0^a \alpha}$ when $|\alpha|=a$
and the first entry of $\alpha$ is $1$), but the proof requires an
alternative construction. This alternative construction and a proof
in the spirit of that of  Lemma~\ref{lem:generation} yields the
following result which deals at once with all the outstanding
cases.
\begin{lemma}
\label{lem:generation2}
Let $\mathcal{P}_{0^a \alpha}$, for $|\alpha|\in \{a,a-1\}$, be a quasi-loop-connected path on vertices $[n]$  (that is not of the form $\mathcal{P}_{0^a 1^b \alpha}$ with $|\alpha|=a$). Let $\mathbb{A}$ be the algebra specified by $\mathrm{id\mbox{-}Pol}(\mathcal{P}_{0^a \alpha})$. For each $m$, $\mathbb{A}^m$ is generated from the $2n+2$ $m$-tuples $(1,1,\ldots,1),$ $(2,2,\ldots,2),(n,1,\ldots,1),$ $(1,n,\ldots,1),$ $\ldots, (1,1,\ldots,n)$,$(n,2,\ldots,2), (2,n,\ldots,2), \ldots, (2,2,\ldots,n)$.
\end{lemma}

We remark that if we were not in the idempotent situation (\mbox{i.e.}
without constants in the structure) then the lemmas could have been
proved from observations about the so-called Q-core
\cite{DBLP:conf/cp/MadelaineM12} via the main result of
\cite{LICS2008} (see Application~\ref{app:FairmontHotelTrick}).


\medskip

\subsubsection{Partially reflexive paths with the EGP}

\

By induction on the arity, we prove the following.
\begin{lemma}
\label{lem:cati}
Let $\alpha$ be any sequence of zeros and ones. All idempotent polymorphisms of $\mathcal{P}_{10\alpha01}$ are projections.
\end{lemma}

This will suffice to derive EGP for all non-quasi loop connected
graphs as we will be able to pinpoint a suitable copy of
$\mathcal{P}_{10\alpha01}$ in all such graphs.
%
%
But first we need to appeal to another ingredient, namely
the well-known Galois correspondence
$\mathrm{Inv}(\mathrm{sPol}(\mathcal{B})) = \langle \mathcal{B}
\rangle_{\mathrm{pH}}$ holding for finite structures $\mathcal{B}$
\cite{BBCJK}, which can be used to derive the following.

\begin{corollary}\label{cor:ppfuniondisguised}
Suppose $\mathbb{A}=\mathrm{id}\mbox{-}\mathrm{Pol}(\mathcal{B})$, for some finite structure $\mathcal{B}$, and $\Gamma$ is a generating set for $\mathbb{A}^m$. Let $\phi(v_1,\ldots,v_m)$ be a formula from $\langle \mathcal{B} \rangle_{\mathrm{pH}}$. If $\mathcal{B} \models \phi(x_1,\ldots,x_m)$ for all $(x_1,\ldots,x_m) \in \Gamma$, then  $\mathcal{B} \models \forall v_1,\ldots,v_m \ \phi(v_1,\ldots,v_m)$.
\end{corollary}



We are now ready to conclude our proof of the PGP/EGP gap for
\mbox{p.r.} paths and establish EGP for the remaining cases.
\begin{proposition}
  \label{prop:pathsEGP}
  Let $\mathcal{G}$ be a \mbox{p.r.} path that is not quasi-loop
  connected. Then $\mathrm{id\mbox{-}Pol}(\mathcal{G})$ has the EGP.
\end{proposition}
\begin{proof}
Number the vertices of $\mathcal{G}$ left-to-right over $[n]$ and let $p$ be the leftmost loop and let $q$ be the rightmost loop. Since $\mathcal{G}$ is not quasi-loop connected, $p$ will be to the left of the centre and $q$ will be to the right of centre. Let $\mu$ be $\max \{p,n-q,\lfloor \frac{q-p-1}{2} \rfloor\}$. Let $P$ and $Q$ be the sets of vertices at distance $\leq \mu$ from $p$ and $q$, respectively.

A word $\tau \in ((P\setminus Q)\cup(Q\setminus P))^m$ is a \emph{cousin} of a word $\sigma \in \{p,q\}^m$ if $\tau$ can be obtained by some local substitutions of $p \mapsto x \in (P\setminus Q)$ and $q \mapsto y \in (Q\setminus P)$.
A word $\tau \in G^m$ is a \emph{friend} of a word $\sigma \in \{p,q\}^m$ if $\tau$ can be obtained by some local substitutions of $p \mapsto 1,\ldots,p,\ldots,p+\mu$ and $q \mapsto q-\mu,\ldots,q,\ldots,n$.  The relations friend and cousin are symmetric. If $\max\{p,n-q\} > q-p-1$ then a situation can arise in which all words $\{p,q\}^m$ are friends of each other (this will not be a problem). However, it is not hard to see that every word in $G^m$ has a friend in $\{p,q\}^m$ and one can walk to this friend pointwise in at most $\mu$ steps. Further, 
\[
  (\dagger)
  \left\{\hspace*{-.5cm}
  \begin{array}{l}
    \begin{minipage}[l]{.9\columnwidth}
      \begin{compactitem}
      \item[.] each word in $((P\setminus Q)\cup(Q\setminus P))^m$ has a unique cousin in $\{p,q\}^m$; and,
      \item[.] every word in $G^m\setminus ((P\setminus Q)\cup(Q\setminus P))^m$ has more than one friend in $\{p,q\}^m$.
      \end{compactitem}
    \end{minipage}
  \end{array}
\right.
\]
Note that it is possible that $G^m\setminus ((P\setminus Q)\cup(Q\setminus P))^m$ is empty. So let $m$ be given and suppose there exists a generating set $\Gamma$ for $G^m$ of size $< 2^m$. It follows from $(\dagger)$ that, for some $\tau \in \{p,q\}^m$, $\Gamma$ omits $\tau$ and all of $\tau$'s cousins (though it may contain some of $\tau$'s non-cousin friends). 
We will prove that $\Gamma$ does not generate $G^m$, by assuming
otherwise and reaching a contradiction using Corollary~\ref{cor:ppfuniondisguised}. Let $R_\Gamma$ be the subset of  $\{p,q\}^m$ induced by $\{p,q\}^m\setminus \{\tau\}$. Note 
\[ (*) \ \ \mbox{that every element $\sigma \in \Gamma$ has a friend in $R_\Gamma$.} \]
Note also that $R_\Gamma$ is pp-definable since $\{p,\ldots,q\}$ is pp-definable and all polymorphisms of the induced sub-structure given by$\{p,\ldots,q\}$ are projections (this was Lemma~\ref{lem:cati}). 

Consider the pH-formula  $\phi(x_1,\ldots,x_n):=$
\begin{multline*}
  \exists x^1_1,\ldots,x^{\mu-1}_1 ,\ldots
  \ldots ,\exists x^1_n,\ldots,x^{\mu-1}_n 
  R_\Gamma(x'_1,\ldots,x'_n)   \wedge \\
  \Bigl(\bigwedge_{i \in [n]} E(x_i,x^1_i) \wedge E(x^1_i,x^2_i)
    \wedge \ldots \\
    \wedge E(x^{\mu-2}_i,x^{\mu-1}_i) \wedge
    E(x^{\mu-1}_i,x^{\mu-1}_i) \Bigr).\label{eq:1}
\end{multline*}

The sentence $\forall x_1,\ldots,x_n \ \phi(x_1,\ldots,x_n)$ is false and can be witnessed as false by taking $\overline{x}$ to be that word in $\{1,n\}^m$ derived from $\tau$ by substituting $p \mapsto 1$ and $q \mapsto n$. However, consider now that $\phi(y_1,\ldots,y_n)$ is true for all $(y_1,\ldots,y_n) \in \Gamma$, precisely because of property $(*)$, \mbox{i.e.} when $(x_1,\ldots,x_n)$ is evaluated as $\sigma$, choose $(x^{\mu-1}_1,\ldots,x^{\mu-1}_n)$ to be evaluated as $\sigma$'s friend in $R_\Gamma$.
\end{proof}

\subsection{Semicomplete digraphs}

Recall that a digraph $\mathcal{G}$ is \emph{semicomplete} if it is
irreflexive and for each $x \neq y \in G$ we have either $E(x,y)$ or
$E(y,x)$, or both. We will often abuse of the substantive and speak of
semicompletes rather than semicomplete graphs. If we always have
precisely one of $E(x,y)$ or $E(y,x)$, then the digraph is
additionally a \emph{tournament}. In a digraph, a \emph{source}
(resp., \emph{sink}) is a vertex  of in-degree (resp., out-degree)
zero. A digraph is  \emph{smooth} if it has neither a source nor a
sink. For a digraph  $\mathcal{G}$ we define $\mathcal{G}^+$ to be
$\mathcal{G}$ augmented with a  new sink to which all other vertices have a directed edge.
%
%
Let $y^{-}$ be the set $\{x \in G : E(x,y)\in \mathcal{G}\}$ and
$y^{+}$ be the set $\{x \in G : E(y,x)\in \mathcal{G}\}$.
%
In the sequel we use the notation $x^{i'}_j$ to indicate the prime of $x^i_j$ (i.e., the prime does not modify just the $i$).

The main result of this section is the gap theorem stated as Theorem~\ref{thm:sc-hauptsatz}.
\begin{proof}[Proof of Theorem~\ref{thm:sc-hauptsatz}]
  The PGP cases follow from Propositions~\ref{prop:sc-1} and
  \ref{prop:sc-2}. The EGP cases follow from
  Corollary~\ref{cor:sc-EGP}.
\end{proof}

\subsubsection{Semicomplete graphs with the PGP}

\


\begin{proposition}
  \label{prop:sc-1}
  Let $\mathcal{G}$ be a semicomplete graph with exactly one cycle and
  either a source or a sink, or none, then $\mathrm{id\mbox{-}Pol}(\mathcal{G})$ has the PGP.
\end{proposition}
\begin{proof}
  If $\mathcal{G}$ has neither source nor sink, then it is either the
  directed $3$-cycle $\mathcal{DC}_3$ or $\mathcal{K}_2$. Let $\mathbb{A}:=$id-Pol$(\mathcal{DC}_3)$ or id-Pol$(\mathcal{K}_2)$. Both of
  these have the dual discriminator for a polymorphism which
  witnesses, for each $a$ in the domain, that 
  $\mathbb{A}^m$ can be generated from tuples, for all $x \in A$, of the form $(a,a,\ldots,a)$, $(x,a,\ldots,a)$, $(a,x,\ldots,a)$, \ldots,  $(a,a,\ldots,x)$ (this latter appears in \cite{hubie-sicomp}). 

  Let us suppose $\mathcal{G}$ has a sink but no source (the alternative being a symmetric proof). Then $\mathcal{G}$ was built from $\mathcal{DC}_3$ or $\mathcal{K}_2$ by the iterative addition of sinks $t_1,\ldots,t_k$, where $t_k$ is the sink of $\mathcal{G}$. Define $f(x,y,z)$ to be the ternary operation on $\mathcal{G}$ that acts as dual discriminator in the subgraph $\mathcal{DC}_3$ or $\mathcal{K}_2$ and returns the element $t_i$ with the highest index $i$ whenever the triple $(x,y,z)$ contains an element from $\{t_1,\ldots,t_k\}$.
  It is straightforward to verify that $f$ is a polymorphism of $\mathcal{G}$. Further, it is a Hubie polymorphism as is witnessed by any element $z$ in the subgraph $\mathcal{DC}_3$ or $\mathcal{K}_2$; that is $f(z,G,G)=f(G,z,G)=f(G,G,z)= G.$
 The result follows from \cite{hubie-sicomp} (that we will quote as Lemma~\ref{lem:ChensLemma}).
\end{proof}
\begin{proposition}\label{prop:sc-2}
  Let $\mathcal{G}$ be a semicomplete graph with both a source and a sink, then $\mathrm{id\mbox{-}Pol}(\mathcal{G})$ has the PGP.
\end{proposition}
\begin{proof}
We will give a Hubie polymorphism of $\mathcal{G}$ whereupon the result follows from \cite{hubie-sicomp} (that we will quote as Lemma~\ref{lem:ChensLemma}).

Let $x,y,z$ be elements of $\mathcal{G}$ distinct from
$s$ and $t$ which are the source and sink, respectively, of
$\mathcal{G}$.
Define the ternary operation $f$ so that $f(\{\{x,s,t\}\})=x$ (we use multiset
notation to indicate any coordinate permutation) extended as a
projection on its first coordinate otherwise (e.g. $f(s,t,s)=f(s,t,t)=s$ and $f(x,y,z)=x$).
It is easy to see this is a polymorphism, once one notes that in
$\mathcal{G}^3$ all vertices of the form $\{\{x,s,t\}\}$ are
isolated. Furthermore, $f$ is a Hubie operation in both the single
elements $s$ and $t$.
\end{proof}

We will shortly need to talk about variables that are indexed individually
over two dimensions and use overbar to denote columns (top index vary)
and underbar to denote rows (bottom index vary).
Suppose id-Pol($\mathcal{A}$) has the $f(m)$-GP. Then we are saying, for each $m \in \mathbb{N}$, that there exist $k=f(m)$ tuples $\overline{x}_1=(x^1_1, x^2_1, \ldots, x^m_1)$, \ldots,  $\overline{x}_k=(x^1_k, x^2_k, \cdots, x^m_k)$ so that, for each $\underline{y}=(y^1, y^2, \ldots, y^m)$ there is a $k$-ary polymorphism $f_{\underline{y}}$ of $\mathcal{A}$ so that 
\[ \underline{y}=(y^1, y^2, \ldots, y^m)=(f_{\underline{y}}(x^1_1,\ldots,x^1_k),\ldots,f_{\underline{y}}(x^m_1,\ldots,x^m_k)).\]
This can be presented by the following picture for $f:=f_{\underline{y}}$,
\[
\begin{array}{cccc}
f & f & \cdots & f \\
\frown & \frown & \cdots & \frown \\
x^1_1 & x^2_1 & \cdots & x^m_1 \\
x^1_2 & x^2_2 & \cdots & x^m_2 \\
\vdots & \vdots & & \vdots \\
x^1_k & x^2_k & \cdots & x^m_k \\
\smile & \smile & \cdots & \smile \\
\| & \| & & \| \\
y^1 & y^2 & \cdots & y^m \\
\end{array}
\] 
which indicates that $f$ is a homomorphism from $(\mathcal{A}^k;\underline{x}^1,\ldots,\underline{x}^m)$ to  $(\mathcal{A};y^1,\ldots,y^m)$. It follows of course that all pp-formulas that are true on  $(\mathcal{A}^k;\underline{x}^1,\ldots,\underline{x}^m)$ are also true on  $(\mathcal{A};y^1,\ldots,y^m)$.

The following well-known model-theoretic lemma is in some sense trivial for finite structures.
\begin{lemma}
\label{lem:basic}
Let $\mathcal{A}$ and $\mathcal{B}$ be finite structures. If all pp-sentences that are true $(\mathcal{A};a_1,\ldots,a_m)$ are true on $(\mathcal{B};b_1,\ldots,b_m)$, then there is a homomorphism $f$ from $\mathcal{A}$ to $\mathcal{B}$ so that $f(a_j)=(b_j)$ for each $j \in [m]$.
\end{lemma}

\subsubsection{Semicompletes with more than one cycle but without
  sources}
It is known from \cite{ICALP2014} that a smooth semicomplete digraph $\mathcal{H}$  with more than one cycle has only essentially unary polymorphisms, since these are also cores we can immediately say in this case that id-Pol$(\mathcal{H})$ has the EGP. What remains is to classify semicompletes with more than one cycle but without sources, and semicompletes with more than one cycle but without sinks. These situations are symmetric so we will address directly only the former. We begin with some simple results.
\begin{lemma}
Let $\mathcal{G}$ be a digraph. Let $\mathrm{id\mbox{-}Pol}(\mathcal{G}^{++})$ have the $f(m)$-GP, for some $f(m)$. Then  $\mathrm{id\mbox{-}Pol}(\mathcal{G}^{+})$ has the $f(m)$-GP.
\end{lemma}
\begin{proof}
Let $t$ be the sink in $\mathcal{G}^{++}$ and let $t'$ be the sink in $\mathcal{G}^{+}\subseteq \mathcal{G}^{++}$. Let $m$ be given and set $k=f(m)$. Let $\overline{x}_1=(x^1_1, x^2_1, \ldots, x^m_1)$, \ldots,  $\overline{x}_k=(x^1_k, x^2_k, \cdots, x^m_k)$ be a set of generators for id-Pol$(\mathcal{G}^{++})$. Set $\overline{x}'_1$, \ldots, $\overline{x}'_k$ to be the tuples obtained from $\overline{x}_1$, \ldots, $\overline{x}_k$ by substituting $t$ by $t'$ and leaving everything else unchanged. We claim that $\overline{x}'_1$, \ldots, $\overline{x}'_k$ is a set of generators for id-Pol$(\mathcal{G}^{+})$. To prove this then, let $\underline{y}=(y^1, y^2, \ldots, y^m) \in {(G^+)}^m$ be given. We need to prove there is $f' \in \mathrm{id}\mbox{-}\mathrm{Pol}(\mathcal{G}^+)$ so that we have the following.
\[
\begin{array}{cccc}
f' & f' & \cdots & f' \\
\frown & \frown & \cdots & \frown \\
x^{1'}_1 & x^{2'}_1 & \cdots & x^{m'}_1 \\
x^{1'}_2 & x^{2'}_2 & \cdots & x^{m'}_2 \\
\vdots & \vdots & & \vdots \\
x^{1'}_k & x^{2'}_k & \cdots & x^{m'}_k \\
\smile & \smile & \cdots & \smile \\
\| & \| & & \| \\
y^1 & y^2 & \cdots & y^m \\
\end{array}
\] 
Let $1,\ldots,n,n+1,n+2$ enumerate the elements of $\mathcal{G}^{++}$ with $\mathcal{G}^+$ being induced on the subset $\{1,\ldots,n,n+1\}$. For $i \in [n+2]$, let $i^k$ denote the $k$-tuple of $i$s. 

By Lemma~\ref{lem:basic}, it is sufficient to show that all pp-formulas that are true on  $((\mathcal{G}^+)^k;1^k,\ldots,(n+1)^k, \underline{x}^{1'},\ldots,\underline{x}^{m'})$ are also true on  $(\mathcal{G}^+;1,\ldots,n+1,y^1,\ldots,y^m)$. 

Let $\phi=\exists \ \overline{w} \phi(\overline{w},\overline{v})$ be a
pp-formula that is true on
$((\mathcal{G}^+)^k;1^k,\ldots,(n+1)^k,\underline{x}^{1'},\ldots,\underline{x}^{m'})$,
that is, for each $j \in [k]$, it is true on
$(\mathcal{G}^+;1,\ldots,n+1,x^{1'}_j,\ldots,x^{m'}_j)$. Let
$\overline{w}_j$ be the witnesses for the existential variables of $\phi$ on this latter structure. Since for all $x \in G^{++}$ we have $E(x,t')$ implies $E(x,t)$, we deduce that $\phi$ is also true on $(\mathcal{G}^{++};1,\ldots,n+1,x^1_j,\ldots,x^m_j)$, using the same witnesses $\overline{w}_0$. Now it follows from $f_{\underline{y}}$ that $\phi$ is true on $(\mathcal{G}^{++};1,\ldots,n+1,y^1,\ldots,y^m)$, by mapping the tuples $\overline{w}_0,\ldots,\overline{w}_k$ under $f_{\underline{y}}$ to obtain the witness for $\overline{w}$ in $(\mathcal{G}^{++};1,\ldots,n+1,y^1,\ldots,y^m)$. But, the idempotent $f_{\underline{y}}$ preserves the set $\{1,\ldots,n,n+1\}$, which is pp-definable in $\mathcal{G}^{++}$, so this shows that the same witnesses show $\phi$ is also true on $(\mathcal{G}^{+};1,\ldots,n+1,x^1_j,\ldots,x^m_j)$. The result follows.
\end{proof}
\begin{corollary}
Let $\mathcal{G}$ be a digraph. If $\mathrm{id\mbox{-}Pol}(\mathcal{G}^{+})$ has the EGP then so does $\mathrm{id\mbox{-}Pol}(\mathcal{G}^{++})$.
\end{corollary}
Let $\mathcal{G}$ be a semicomplete digraph with more than one cycle and no source. We say $\mathcal{G}$ has the \emph{Novi Sad property} if there exist vertices $p,q \in G$ so that
\begin{itemize}
\item for all $v \in  G$ there is the edge $E(v,p)$ or $E(v,q)$.
\end{itemize}
Note that the Novi Sad property implies a double edge between $p$ and
$q$, hence this fails on all tournaments. Importantly for our uses, on 
irreflexive graphs this property implies that (picking $p':=q$ and $q':=p$):
\begin{itemize}
\item exists $p' \in  G$ so that $E(p',p)$ but not $E(p',q)$,
\item exists $q' \in  G$ so that $E(q',q)$ but not $E(q',p)$.
\end{itemize}
The Novi Sad property does not feature in \cite{ICALP2014}.

\noindent \textbf{Specific results imported from \cite{ICALP2014}}. We now need to borrow some definitions and results from \cite{ICALP2014}. In that paper the authors usually refer to Pol instead of id-Pol, but the the objects are always cores expanded by constants, so the two coincide.
\begin{definition}[Definition~6 in \cite{ICALP2014}]\label{def:1}
Let $\str{G}$ be a directed graph. We define the relation $\preceq_{\str G}$ on $V$ by $x\preceq_{\str G} y$ iff $x^-\subseteq y^-$.
\end{definition}
\begin{proposition}[Proposition~9 in \cite{ICALP2014}]
\label{order}
Assume that $\str G$ is semicomplete. Then $\preceq_{\str G}$ is a partial order, $\preceq_{\str G}$ has the largest element $t$ iff $t$ is a sink, and dually for least elements and sources.
\end{proposition}
\begin{definition}[Definition~7 in \cite{ICALP2014}]\label{def:2}
Let ${\str G}$ be a digraph. We define the partition of the vertex set $V$ into $V_{min}^{\str G}$, $V_{max}^{\str G}$, $V_{both}^{\str G}$ and $V_{none}^{\str G}$ so that all vertices in $V_{max}^{\str G}$ are maximal, but not minimal, in the order $\preceq_\str g$, all vertices in $V_{min}^\str g$ are minimal, but not maximal, in the order $\preceq_\str g$, all vertices in $V_{both}^\str g$ are both minimal and maximal in the order $\preceq_\str g$, while vertices in $V_{none}^\str g$ are neither minimal nor maximal in the order $\preceq_\str g$. When the digraph $\str g$ is understood, we will omit the superscript $^\str g$.
\end{definition}
\begin{definition}[Definition~8 in \cite{ICALP2014}]\label{sofg}
Let $\str g$ be a digraph. We define the irreflexive digraph $\str S(\str g)$ by:
\begin{enumerate}
\item For all $x,y\in V_{max}\cup V_{both}$, $(x,y), (y,x) \in E(\str S(\str g))$,
\item For all $x,y\in V_{min}$, $(x,y), (y,x) \in E(\str S(\str g))$,
\item For all $x,y\in V_{none}$, $(x,y) \in E(\str S(\str g))$ iff $(x,y) \in E(\str g)$.
\item For all $x\in V_{min}$ and $y\in V_{none}\cup V_{max}$, $(x,y) \in E(\str S(\str g))$, but not $(y,x) \in E(\str S(\str g))$,
\item For all $x\in V_{none}$ and $y\in V_{max}$, $(x,y) \in E(\str S(\str g))$, but not $(y,x) \in E(\str S(\str g))$,
\item For all $x\in V_{both}$ and $y\in V_{none}\cup V_{min}$, $(x,y) \in E(\str S(\str g))$, but not $(y,x) \in E(\str S(\str g))$.
\end{enumerate}
\end{definition}
\begin{proposition}[Proposition~10 in \cite{ICALP2014}]\label{sofsofg}
$V_{min}^{\str s(\str g)}=V_{min}^\str g$, $V_{max}^{\str s(\str g)}=V_{max}^\str g$, $V_{both}^{\str s(\str g)}=V_{both}^\str g$ and $V_{none}^{\str s(\str g)}=V_{none}^\str g$. Consequently, $\str s(\str s(\str g))=\str s(\str g)$.
\end{proposition}
\begin{corollary}[Corollary~6 in \cite{ICALP2014}]\label{sofgreduction}
Let $\str g$ be a smooth semicomplete digraph which is not a cycle. Then $\mathrm{id\mbox{-}Pol}(\str g^+)\subseteq \mathrm{id\mbox{-}Pol}(\str s(\str g)^+)$.
\end{corollary}

\vspace{.1cm}
\noindent \textbf{Applications of results imported from \cite{ICALP2014}}. 
\begin{theorem}
Let $\mathcal{G}$ be a smooth semicomplete with more than one cycle. There exists a smooth semicomplete with more than one cycle $\mathcal{H}$ so that $\mathrm{id\mbox{-}Pol}(\mathcal{G}^+) \subseteq \mathrm{id\mbox{-}Pol}(\mathcal{H}^+)$ and $\mathcal{H}^+$ has the Novi Sad property.
\end{theorem}
\begin{proof}
Note that  $|V_{both}^\str g\cup V_{max}^\str g|\geq 2$,  so we can apply Corollary \ref{sofgreduction}, choosing $\mathcal{H}=\str s(\str g)$, with $p\neq q$ chosen as follows:
 If $V_{max}^\str g=\emptyset$, this implies that $V_{min}^\str g=V_{none}^\str g=\emptyset$ and $V=V_{both}^\str g$, and $p, q$ can be chosen arbitrarily;
 If $V_{max}^\str g\neq \emptyset$, then  we choose $p\in V_{max}^\str g$ and $q\in V_{max}^\str g\cup V_{both}^\str g$. Then $ (p,q), (q, p)\in E(\str s(\str g))$ (and this graph has no loops),  and there is an edge from all vertices of $\str s(\str g)$, except $p$,  to $p$.
%
%
%
\end{proof}

\vspace{.1cm}
\noindent \textbf{Main EGP result for semicompletes}.

\begin{proposition}
\label{prop:sc-EGP}
Let $\mathcal{G}$ be a semicomplete digraph with more than one cycle, no source, and the Novi Sad property. Then $\mathrm{id\mbox{-}Pol}(\mathcal{G})$ has the EGP.
\end{proposition}
\begin{proof}
Let $p$ and $q$, together with $p'$ and $q'$, be as guaranteed to
exist by the Novi Sad property. 
Let $U$ be the unary relation specifying the domain of the smooth
semicomplete digraph with more than one cycle which is obtained from
$\mathcal{G}$ by removing sinks repeatedly.

A word $\tau \in G^m$ is said to be a \emph{sub-predecessor} of a word
$\sigma \in \{p,q\}^m$ if $\tau$ can be obtained by some local
substitutions of $p \mapsto x\in p^{-}$ and $q \mapsto x \in
q^{-}$. If $\tau$ is a sub-predecessor of $\sigma$ then we may say
$\sigma$ is a \emph{sub-successor} of $\tau$. Note that every word
$\tau \in G^m$ has a sub-successor in  $\sigma \in \{p,q\}^m$, by the Novi Sad property.
%
A word $\tau \in G^m$ is said to be a \emph{predecessor} of a word $\sigma \in \{p,q\}^m$ if $\tau$ can be obtained by some local substitutions of $p \mapsto x \in p^{-}\setminus q^{-}$ and $q \mapsto x \in q^{-}\setminus p^{-}$. If $\tau$ is a predecessor of $\sigma$ then we may say $\sigma$ is a \emph{successor} of $\tau$. Note that predecessor (resp., successor) imply sub-predecessor (resp., sub-successor).
Now,
\[
  (\dagger)
  \left\{\hspace*{-.5cm}
  \begin{array}{l}
    \begin{minipage}[l]{.9\columnwidth}
      \begin{compactitem}
      \item[.] each word in  $((p^{-}\setminus q^{-}) \cup (q^{-}\setminus p^{-}))^m$ has a unique successor in $\{p,q\}^m$; and
      \item[.] every word in $G^m\setminus ((p^{-}\setminus q^{-}) \cup (q^{-}\setminus p^{-}))^m$ has more than one sub-successor in $\{p,q\}^m$.
      \end{compactitem}
    \end{minipage}
  \end{array}
\right.
\]
In analogy to the proof of Proposition~\ref{prop:pathsEGP}, predecessor/ successor play the role of cousin and sub-predecessor/ sub-successor play the role of friend. 

Let $m$ be given and suppose there exists a generating set $\Gamma$ for $G^m$ of size $< 2^m$. It follows from $(\dagger)$ that, for some $\tau \in \{p,q\}^m$, $\Gamma$ omits $\tau$ and all of $\tau$'s predecessors.

We will prove that $\Gamma$ does not generate $G^m$, by assuming otherwise and reaching a contradiction. Let $R_\Gamma$ be the subset $\{p,q\}^m\setminus \{\tau\}$. Note 
\[ (*) \ \ \mbox{that every $\sigma \in \Gamma$ has a sub-successor in $R_\Gamma$.} \]
Note also that $R_\Gamma$ is pp-definable since $U$ is pp-definable and all polymorphisms of the sub-structure induced by $U$ are projections (see \cite{ICALP2014}). 

Consider the pH-formula $\phi(x_1,\ldots,x_n):=$
\[\exists x'_1,\ldots,x'_n \ \left(\bigwedge_{i \in [n]} E(x_i,x'_i) \right) \wedge R_\Gamma(x'_1,\ldots,x'_n).\]
The sentence $\forall x_1,\ldots,x_n \ \phi(x_1,\ldots,x_n)$ is false and can be witnessed as false by taking $\overline{x}$ to be that word in $\{p',q'\}^m$ derived from $\tau$ by substituting $p \mapsto p'$ and $q \mapsto q'$.  However, consider now that $\phi(y_1,\ldots,y_n)$ is true for all $(y_1,\ldots,y_n) \in \Gamma$, precisely because of property $(*)$, \mbox{i.e.} when $(x_1,\ldots,x_n)$ is evaluated as $\sigma$, choose $(x'_1,\ldots,x'_n)$ to be evaluated as $\sigma$'s sub-successor in  $R_\Gamma$.
\end{proof}

\begin{corollary}
\label{cor:sc-EGP}
Let $\mathcal{G}$ be a semicomplete digraph with more than one cycle and either no source or no sink.  Then $\mathrm{id\mbox{-}Pol}(\mathcal{G})$ has the EGP.
\end{corollary}
\begin{proof}
From \cite{ICALP2014} we know that semicomplete digraphs $\mathcal{H}$ with more than one cycle and neither a source nor a sink (smooth) have only essentially unary polymorphisms. It follows of course that id-Pol$(\mathcal{H})$ has the EGP. The result now follows from Proposition~\ref{prop:sc-EGP} (and its symmetric dual).
\end{proof}
\section{\BarnyApprovedTitleForOurStuffText}
\label{sec:from-pgp-complexity}
Throughout this section, we shall be concerned with a relational
structure $\mathcal{A}$ over a finite set $A$ of size $n$. In the few
cases when we will require $\mathcal{A}$ to have specific constants,
we shall state it explicitly. 
\subsection{Games, adversaries and reactive composition}
\label{sec:games-adversaries-reactivecomposition}
%
We recall some terminology due to
Chen~\cite{hubie-sicomp,AU-Chen-PGP}, for his natural adaptation of
the model checking game to the context of pH-sentences. 
We shall not need to explicitly play these games but only to handle
strategies for the existential player. 
An \emph{adversary} $\mathscr{B}$ of length $m\geq 1$ is an $m$-ary
relation over $A$. 
When $\mathscr{B}$ is precisely the set
$B_{1}\times B_{2} \times \ldots \times B_{m}$ for some non-empty
subsets $B_1,B_2,\ldots,B_m$ of $A$, we speak of a \emph{rectangular
  adversary}.
Let $\phi$ have universal variables $x_1,\ldots,x_m$ and quantifier-free
part $\psi$. We write $\mathcal{A}\models \phi_{\restrict\mathscr{B}}$ and say that
\emph{the existential player has a winning strategy in the
$(\mathcal{A},\phi)$-game against adversary $\mathscr{B}$} iff there
exists a set of Skolem functions $\{\sigma_x : \mbox{`$\exists x$'}
\in \phi \}$ such that for any assignment $\pi$ of the universally
quantified variables of $\varphi$ to $A$, where
$\bigl(\pi(x_1),\ldots,\pi(x_m)\bigr) \in \mathscr{B}$, the map $h_\pi$ 
is a homomorphism from $\mathcal{D}_\psi$ (the canonical database) to $\mathcal{A}$, where
$$h_\pi(x):=
\begin{cases}
  \pi(x) & \text{, if $x$ is a universal variable; and,}\\
  \sigma_x(\left.\pi\right|_{Y_x})& \text{, otherwise.}\\
\end{cases}
$$
(Here, $Y_x$ denotes the set of universal variables preceding $x$ and $\left.\pi\right|_{Y_x}$ the restriction of $\pi$ to $Y_x$.)
Clearly, $\mathcal{A} \models \phi$ iff the existential player has a winning strategy in the
$(\mathcal{A},\phi)$-game against the so-called \emph{full
  (rectangular) 
  adversary} $A\times A \times \ldots \times A$ (which
we will denote hereafter by $A^m$).
We say that an adversary $\mathscr{B}$ of length $m$ \emph{dominates}
an adversary $\mathscr{B}'$ of length $m$ when $\mathscr{B}'\subseteq
\mathscr{B}$. Note that $\mathscr{B}'\subseteq  \mathscr{B}$ and
$\mathcal{A}\models \phi_{\restrict\mathscr{B}}$ implies
$\mathcal{A}\models \phi_{\restrict\mathscr{B}'}$. 
We will also
consider sets of adversaries of the same length, denoted
by uppercase greek letters as in $\Omega_m$; and, sequences thereof, which we
denote with bold uppercase greek letters as in
$\mathbf{\Omega}=\bigl(\Omega_m\bigr)_{m \in \integerset}$. 
We will write
$\mathcal{A}\models \phi_{\restrict\Omega_m}$ to denote that
$\mathcal{A}\models  \phi_{\restrict\mathscr{B}}$ holds for every
adversary $\mathscr{B}$ in $\Omega_m$. 
We call \emph{width} of $\Omega_m$ and write $\mathrm{width}(\Omega_m)$ for $\sum_{\mathscr{B}\in
  \Omega_m} |\mathscr{B}|$.
We say that $\mathbf{\Omega}$ is \emph{polynomially bounded} if
there exists a polynomial $p(m)$ such that for every $m\geq 1$,
$\mathrm{width}(\Omega_m) \leq p(m)$. We say that $\mathbf{\Omega}$ is
\emph{effective} if there exists a polynomial $p'(m)$ and an algorithm that outputs
$\Omega_m$ for every $m$ in total time
$p'(\mathrm{width}(\Omega_m))$.
%
%

Let $f$ be a $k$-ary operation of $\mathcal{A}$ and $\mathscr{A},\mathscr{B}_1,\ldots,\mathscr{B}_k$ be adversaries of length $m$.
We say that $\mathscr{A}$ is \emph{reactively composable} from the
adversaries $\mathscr{B}_1,\ldots,\mathscr{B}_k$ via $f$, and we write $\mathscr{A} \reactivelycomposable f(\mathscr{B}_1,\ldots,\mathscr{B}_k)$ iff there exist partial functions $g^j_i:A^i \to A$ for every $i$ in $[m]$ and every $j$ in $[k]$ such that, for every tuple $(a_1,\ldots,a_m)$ in adversary $\mathscr{A}$ the following holds.
\begin{compactitem}
\item for every $j$ in $[k]$, the values 
  $g^j_1(a_1), g^j_2(a_1,a_2),$ $\ldots , g^j_m(a_1,a_2,\ldots,a_m)$ are defined and the tuple $\bigl(g^j_1(a_1), g^j_2(a_1,a_2), \ldots, g^j_m(a_1,a_2,\ldots,a_m)\bigr)$ is in adversary $\mathscr{B}_j$; and,
\item for every $i$ in $[m]$, $a_i =f\bigl(g^1_i(a_1,a_2,\ldots,a_i),$
  $g^2_i(a_1,a_2,\ldots,a_i),\ldots,g^k_i(a_1,a_2,\ldots,a_i))$.
\end{compactitem}
We write $\mathscr{A} \reactivelycomposable
\{\mathscr{B}_1,\ldots,\mathscr{B}_k\}$ if there exists a $k$-ary
operation $f$ such that $\mathscr{A} \reactivelycomposable f(\mathscr{B}_1,\ldots,\mathscr{B}_k)$

\begin{remark}
  We will never show reactive composition by exhibiting a polymorphism $f$ and  partial functions $g^i_j$ that depend on all their arguments. We will always be able to exhibit partial functions that depend only on their last argument. 
\end{remark}
Reactive composition allows to interpolate complete Skolem functions
from partial ones.
\begin{theorem}[{\cite[Theorem 7.6]{AU-Chen-PGP}}]
  \label{thm:hubieReactiveComposition}
  Let $\phi$ be a pH-sentence with $m$ universal variables. Let $\mathscr{A}$ be an adversary and $\Omega_m$ a set of adversaries, both of length $m$.
  
  If $\mathcal{A}\models \phi_{\restrict\Omega_m}$ and  $\mathscr{A} \reactivelycomposable \Omega_m$ then
  $\mathcal{A} \models \phi$.
\end{theorem}
As a concrete example of an interesting sequence of adversaries, consider the adversaries for the notion of
\emph{$p$-collapsibility}, which we introduced in a purely logical fashion in
the introduction. Let $p\geq 0$ be some fixed integer.
For $x$ in $A$, let $\Upsilon_{m,p,x}$ be the set of all 
rectangular 
adversaries of length $m$ with $p$ coordinates that are the set $A$
and all the other that are the fixed singleton $\{x\}$. For
$B\subseteq A$, let $\Upsilon_{m,p,B}$ be the union of $\Upsilon_{m,p,x}$ for all $x$ in $B$.
Let $\mathbf\Upsilon_{p,B}$ be the sequence of adversaries
$\Bigl(\Upsilon_{m,p,B} \Bigr)_{m \in \integerset}$.
We will define a structure $\mathcal{A}$ to be \emph{$p$-collapsible
  from source $B$} iff for every $m$ and for all pH-sentence $\phi$
with $m$ universal variable, $\mathcal{A}\models \phi_{\restrict\Upsilon_{m,p,B}}$ implies   $\mathscr{A} \models \phi$.

\subsection{The $\Pi_2$-case}
\label{sec:pi2}
For a $\Pi_2$-pH sentence, the existential player knows the values of
all universal variables beforehand, and it suffices for her to have a
winning strategy for each instantiation (and perhaps no way to
reconcile them as should be the case for an arbitrary sentence). 
This also means that considering a set of adversaries of same length
is not really relevant in this $\Pi_2$-case as we may as well consider
the union of these adversaries or the set of all their tuples (see
also statement of Corollary~\ref{cor:ppfuniondisguised}).
\begin{lemma}[\textbf{principle of union}]
  \label{lemma:union}
  Let $\Omega_m$ be a set of adversaries of length $m$ and $\phi$ a $\Pi_2$-sentence with $m$ universal variables.
  Let $\mathscr{O}_{\cup\Omega_m}:=\bigcup_{\mathscr{O}\in \Omega_m} \mathscr{O}$
  and $\Omega_{\text{tuples}}:=\{\{t\} | t \in \mathscr{O}_{\cup\Omega_m}\}$. 
  We have the following equivalence.
  $$\mathcal{A}\models \varphi_{\restrict\Omega_m} 
  \quad \iff \quad
  \mathcal{A}\models \varphi_{\restrict\mathscr{O}_{\cup\Omega_m}}
  \quad \iff \quad
  \mathcal{A}\models \varphi_{\restrict\Omega_{\text{tuples}}}$$
\end{lemma}
Let $\mathscr{A}$ be an adversary and $\Omega_m$ a set of adversaries, both of length $m$.
We say that $\Omega_m$ \emph{generates} $\mathscr{A}$ iff for any tuple $t$ in $\mathscr{A}$, there exists a $k$-ary polymorphism $f_t$ of
$\mathcal{A}$ and tuples $t_1,\ldots,t_k$ in $\Omega_{\text{tuples}}$
such that $f_t(t_1,\ldots,t_k)=t$.
We have the following analogue of Theorem~\ref{thm:hubieReactiveComposition}.
\begin{proposition}
  \label{prop:pi2composition}
  Let $\phi$ be a $\Pi_2$-pH-sentence with $m$ universal variables. Let $\mathscr{A}$ be an adversary and $\Omega_m$ a set of adversaries, both of length $m$.
  
  If $\mathcal{A}\models \phi_{\restrict\Omega_m}$ and  $\Omega_m$
  generates $\mathscr{A}$  then
  $\mathcal{A} \models \phi_{\restrict\mathscr{A}}$.
\end{proposition}

We will construct a \emph{canonical $\Pi_2$-sentence to assert that an adversary is
  generating.} Let $\mathscr{O}$ be some adversary of length $m$.
Let $\sigma^{(m)}$ be the signature $\sigma$ expanded with a sequence of $m$ constants. For a map $\mu$ from $[m]$ to $A$, we write $\mu\in \mathscr{O}$ as shorthand for $(\mu(1),\mu(2),\ldots,\mu(m))\in \mathscr{O}$.
For some set $\Omega_m$ of adversaries of length $m$, we consider the following $\sigma^{(m)}$-structure: 
$$\bigotimes_{\mathscr{O}\in \Omega_m} \bigotimes_{\mu \in \mathscr{O}} \mathfrak{A}_{\mu}$$
where the $\sigma^{(m)}$-structure $\mathfrak{A}_{\mu}$ denotes the expansion of $\mathcal{A}$ by $m$ constants as given by the map $\mu$.
Let $\phi_{\Omega_m,\mathcal{A}}$ be the $\Pi_2$-pH-sentence\footnotemark{} created
from the canonical query of the $\sigma$-reduct of this
$\sigma^{(m)}$-structure with the $m$ constants $c_{j}$ becoming
variables $w_{j}$, universally quantified outermost, when all
\emph{constants are pairwise distinct}.
\footnotetext{For two structures $\mathcal{A}$ and $\mathcal{B}$, when
  $\Omega_m$ is $A^m$ and $m$ is $|A|^{B}$, $\mathcal{B}$ models this
  canonical sentence iff  $\QCSP(\mathcal{A}) \subseteq \QCSP(\mathcal{B})$~\cite{LICS2008}}
Otherwise, we will say that $\Omega_m$ is \emph{degenerate}, and not
define the canonical sentence.

Note that adversaries such as $\Upsilon_{m,p,B}$ corresponding to $p$-collapsibility 
are not degenerate for $p>0$, and degenerate for $p=0$.
\begin{proposition}
  \label{thm:characteringPi2}
  Let $\Omega_m$ be a set of adversaries of length $m$ that is not degenerate.
  The following are equivalent.
  \begin{romannum}
  \item for any $\Pi_2$-pH sentence $\psi$, $\mathcal{A}\models
    \psi_{\restrict\Omega_m}$ implies $\mathcal{A}\models \psi$.
    \label{pi2abstracto:logical:pi2}
  \item for any $\Pi_2$-pH sentence $\psi$, $\mathcal{A}\models \psi_{\restrict\mathscr{O}_{\cup\Omega}}$ implies $\mathcal{A}\models \psi$.
  \item for any $\Pi_2$-pH sentence $\psi$, $\mathcal{A}\models \psi_{\restrict\Omega_{\text{tuples}}}$ implies $\mathcal{A}\models \psi$.
  \item $\mathcal{A}\models
    \phi_{\mathscr{O}_{\cup\Omega},\mathcal{A}}$
    \label{pi2abstracto:canonical:pi2}
  \item $\mathcal{A}\models \phi_{\Omega_{\text{tuples}},\mathcal{A}}$
  \item $\Omega_m$ generates $A^m$.\label{pi2abstracto:algebraic:pi2}
  \end{romannum}
\end{proposition}
\subsection{The unbounded case}
\label{sec:unbounded}
Let $n$ denote the number of elements of the structure
$\mathcal{A}$. Let $\mathscr{B}$ be an adversary from
$\Omega_{n\cdot m}$. We will denote by $\text{Proj}\mathscr{B}$ the set of
adversaries of  length $m$ induced by projecting over some arbitrary
choice of $m$ coordinates, one in each block of size $n$; that is
$1\leq i_1\leq n,n+1\leq i_2\leq 2\cdot n, \ldots, n\cdot (m-1)+ 1\leq
i_m\leq n \cdot m$. 
Of special concern to us are \emph{projective sequences of adversaries}
$\mathbf{\Omega}$ satisfying the following  for every $m\geq 1$,
\begin{equation*}
  \forall \mathscr{B} \in \Omega_{n.m} \,\,
  \exists \mathscr{A} \in \Omega_{m} \,\,
  \bigwedge_{\widetilde{\mathscr{B}} \in \text{Proj}\mathscr{B}} 
  \widetilde{\mathscr{B}} \subseteq \mathscr{A} \,\,\,\,
  \text{($m$-\textbf{projectivity})}
\end{equation*}
As an example, consider the adversaries for collapsibility.
\begin{fact} 
  \label{fact:collapsible:adversary:is:projective}
  Let $B \subseteq A$ and $p\geq 0$.
  The sequence of adversaries $\mathbf\Upsilon_{p,B}$ 
  are projective.
\end{fact}
\begin{example}
  For a concrete illustration consider $A=\{0,1,2\}$ (thus $n=3$).
  We illustrate the fact that $\mathbf\Upsilon_{p=2,B=\{0\}}$ is
  projective for $m=4$ and some adversary $\mathscr{B}\in \Omega_{n\cdot
      m}=\Upsilon_{p=2,B=\{0\},3\cdot 4=12}$. Adversaries are depicted
    vertically with horizontal lines separating the blocks.
  $${\scriptsize \begin{array}{c|cccc|c}
    \mathscr{B}\in \Omega_{n\cdot
      m}&\multicolumn{4}{c|}{\text{Proj}\mathscr{B}}& \mathscr{A} \in
    \Omega_{m}\\
    \hline
    \hline
    A &A          &A          &      &\xcancel{A} &  \\
    0 &\xcancel{0}&\xcancel{0}&\ldots&\xcancel{0} & A\\
    0 &\xcancel{0}&\xcancel{0}&      &0           &  \\
    \hline
    0 &0          &0          &      &\xcancel{0}&  \\
    0 &\xcancel{0}&\xcancel{0}&\ldots&\xcancel{0}& 0\\
    0 &\xcancel{0}&\xcancel{0}&      &0          &  \\
    \hline
    0 &0          &0          &      &\xcancel{0}&  \\
    0 &\xcancel{0}&\xcancel{0}&\ldots&\xcancel{0}& 0\\
    0 &\xcancel{0}&\xcancel{0}&      &0          &  \\
    \hline
    0 &0          &\xcancel{0}&      &\xcancel{0}&  \\
    A &\xcancel{A}&A          &\ldots&\xcancel{A}& A\\
    0 &\xcancel{0}&\xcancel{0}&      &0          &  \\
  \end{array}}$$
The adversary $\mathscr{A}$ dominates any adversary obtained by
projecting the original larger adversary $\mathscr{B}$ by keeping a
single position per block. 
\end{example}
We could actually consider w.l.o.g. sequences of \emph{singleton}
adversaries.
\begin{fact}
  If $\mathbf{\Omega}$ is projective then so is the sequence
  $\bigl(\bigcup_{\mathscr{O}\in\Omega_m} \mathscr{O}\bigr)_{m \in  \integerset}$.
\end{fact}

A \emph{canonical sentence for composability for arbitrary pH-sentences} with $m$ universal variables
may be constructed similarly to the canonical sentence for the $\Pi_2$ case, except that it will have $m.n$ universal variables, which we view as $m$ blocks of $n$ variables, where $n$ is the number of elements of the structure $\mathcal{A}$.
Let $\mathscr{O}$ be some adversary of length $m$.
Let $\sigma^{(n\cdot m)}$ be the signature $\sigma$ expanded with a sequence of $n.m$ constants $c_{1,1},\ldots, c_{n,1}, c_{1,2} \ldots, c_{n,2}, \ldots c_{1,m} \ldots, c_{n,m}$. We say that a map $\mu$ from $[n]\times [m]$ to $A$ is \emph{consistent} with $\mathscr{O}$ iff for every $(i_1,i_2,\ldots,i_m)$ in $[n]^m$, the tuple $(\mu(i_1,1),\mu(i_2,2),\ldots,\mu(i_m,m))$ belongs to the adversary $\mathscr{O}$. We write $A^{[n.m]}_{\restrict \mathscr{O}}$ for the set of such consistent maps.
For some set $\Omega_m$ of adversaries of length $m$, we consider the following $\sigma^{(n.m)}$-structure: 
$$\bigotimes_{\mathscr{O}\in \Omega_m} \bigotimes_{\mu \in A^{[n.m]}_{\restrict \mathscr{O}}} \mathfrak{A}_{\mathscr{O},\mu}$$
where the $\sigma^{(n\cdot m)}$-structure $\mathfrak{A}_{\mathscr{O},\mu}$ denotes the expansion of $\mathcal{A}$ by $n.m$ constants as given by the map $\mu$.
Let $\phi_{n,\Omega_m,\mathcal{A}}$ be the $\Pi_2$-pH-sentence created from the canonical query of the $\sigma$-reduct of this $\sigma^{(n.m)}$ product structure with the $n.m$ constants $c_{ij}$ becoming variables $w_{ij}$, universally quantified outermost.
As for the canonical sentence of the $\Pi_2$-case, this sentence is
not well defined if constants are not pairwise distinct, which occurs precisely
for degenerate adversaries. 
\begin{lemma}
  \label{lem:CanonicalSentenceImpliesReactiveComposability}
  Let $\Omega_m$ be a set of adversaries of length $m$ that is not degenerate. Let $\mathcal{A}$ be a structure of size $n$.
  If $\mathcal{A}$ models $\phi_{n,\Omega_m,\mathcal{A}}$ then the full adversary $A^m$ is reactively composable from $\Omega_m$.
  That is,
  $\mathcal{A} \models \phi_{n,\Omega_m,\mathcal{A}} \quad \implies \quad A^m \reactivelycomposable \Omega_m$
\end{lemma}
\begin{proof}
  We let each block of $n$ universal variables of the canonical
  sentence $\phi_{n,\Omega_m,\mathcal{A}}$ enumerate the elements of $A$.
  That is, given an enumeration $a_1,a_2,\ldots,a_n$ of $A$, we set $w_{i,j}=a_i$ for every $j$ in $[m]$ and every $i$ in $[n]$.  

  The assignment to the existential variables provides us with a
  $k$-ary polymorphism (the sentence being built as the conjunctive
  query of a product of $k$ copies of $\mathcal{A}$) together with the
  desired partial maps. A coordinate $r$ in $[k]$ corresponds to a
  choice of some adversary $\mathscr{O}$ of $\Omega_m$ and some map
  $\mu_r$ from $[n]\times[m]$ to $A$, consistent with this adversary. The partial map $g^r_\ell:A^\ell\to A$ with $\ell$ in $[m]$ (and $r$ in $[k]$) is given by $\mu_r$ as follows:
  $g^r_\ell(a_{i_1},\ldots,a_{i_\ell})$ depends only on the last coordinate $a_{i_\ell}$ and takes value $\mu(i,\ell)$ if $a_{i_\ell}=a_i$.
  By construction of the sentence and the property of consistency of such $\mu_r$ with the adversary $\mathscr{O}$, these partial functions satisfy the properties as given in the definition of reactive composition.
\end{proof}
\begin{lemma} 
  Let $\mathbf{\Omega}$ be a sequence of sets of adversaries that has
  the $m$-projectivity property for some $m\geq 1$ such that
  $\Omega_{n\cdot m}$ is not degenerate. The following holds.
  \label{lem:AdversariesWinnableOnCanonicalSentence}
  \begin{romannum}
  \item 
    $\mathcal{A}\models \psi_{\restrict \Omega_{\mathbf{n.m}}}, \text{ where } \psi={\phi_{n,\Omega_\mathbf{m},\mathcal{A}}}$
  \item If for every $\Pi_2$-sentence $\psi$ with $m.n$ universal variables, it holds that $\mathcal{A}\models \psi_{\restrict \Omega_{\mathbf{m.n}}}$ implies $\mathcal{A}\models \psi$, then
    $\mathcal{A}\models \phi_{n,\Omega_\mathbf{m},\mathcal{A}}$.
  \end{romannum}
\end{lemma}
\begin{theorem}
  \label{theo:Pi2ToGeneral}
  Let $\mathbf{\Omega}$ be a sequence of sets of adversaries that has
  the $m$-projectivity property for some $m\geq 1$ such that
  $\Omega_{n.m}$ is not degenerate.
  The following chain of implications holds
  $$
  \ref{Pi2ToGeneral:Pi2GoodApproximation}
  \implies
  \ref{Pi2ToGeneral:BigPi2Canonicalsentence}
  \implies 
  \ref{Pi2ToGeneral:ReactiveComposition}
  \implies 
  \ref{Pi2ToGeneral:GoodApproximation}
  $$
  where,
  \begin{romannum}
  \item For every $\Pi_2$-pH-sentence $\psi$ with $m.n$ universal variables, $\mathcal{A}\models \psi_{\restrict \Omega_{m.n}}$ implies $\mathcal{A}\models \psi$.\label{Pi2ToGeneral:Pi2GoodApproximation}
  \item $\mathcal{A}\models \phi_{n,\Omega_m,\mathcal{A}}$.\label{Pi2ToGeneral:BigPi2Canonicalsentence}
  \item $A^m \reactivelycomposable \Omega_m$.\label{Pi2ToGeneral:ReactiveComposition}
  \item For every pH-sentence $\psi$ with $m$ universal variables, $\mathcal{A}\models \psi_{\restrict \Omega_{m}}$ implies $\mathcal{A}\models \psi$.\label{Pi2ToGeneral:GoodApproximation}
  \end{romannum}
\end{theorem}
\begin{proof}
  The first implication holds by the previous lemma (second item of Lemma~\ref{lem:AdversariesWinnableOnCanonicalSentence}, this is the step where we use projectivity).
  The second implication is Lemma~\ref{lem:CanonicalSentenceImpliesReactiveComposability}. The last implication is Theorem~\ref{thm:hubieReactiveComposition}.
\end{proof}
Thus, in the projective case, when an adversary is good enough in the
$\Pi_2$-case, it is good enough in general. This can be characterised
logically via canonical sentences or ``algebraically'' in terms of
reactive composition or the weaker and more usual composition
property (see \ref{abstracto:algebraic:pi2}  below). 
  \begin{theorem}[\textbf{In abstracto}]\label{MainResult:InAbstracto}
    Let $\mathbf{\Omega}=\bigl(\Omega_m\bigr)_{m \in \integerset}$ be a projective sequence of adversaries,
    none of which are degenerate. 
    The following are equivalent.
  \begin{romannum}
  \item For every $m\geq 1$, for every pH-sentence $\psi$ with $m$
    universal variables, $\mathcal{A}\models \psi_{\restrict
      \Omega_{m}}$ implies $\mathcal{A}\models \psi$. 
    \label{abstracto:logical:general}
  \item For every $m\geq 1$, for every $\Pi_2$-pH-sentence $\psi$ with $m$ universal
    variables, $\mathcal{A}\models \psi_{\restrict \Omega_{m}}$
    implies $\mathcal{A}\models \psi$.
    \label{abstracto:logical:pi2}
  \item For every $m\geq 1$, $\mathcal{A}\models
    \phi_{n,\Omega_m,\mathcal{A}}$.
    \label{abstracto:canonical:general}
  \item For every $m\geq 1$, $\mathcal{A}\models
    \phi_{\mathscr{O}_{\cup\Omega},\mathcal{A}}$.
    \label{abstracto:canonical:pi2}
  \item For every $m\geq 1$, $A^m \reactivelycomposable \Omega_m$.
    \label{abstracto:algebraic:general}
  \item For every $m\geq 1$, $\Omega_m$ generates $A^m$.
    \label{abstracto:algebraic:pi2}
  \end{romannum}
\end{theorem}
\begin{remark}
  The above equivalences can be read along two dimensions:
  \begin{center}
    \begin{tabular}[h]{lc|c}
      & general        & $\Pi_2$       \\
      \hline
      logical interpolation &
      \ref{abstracto:logical:general}&~\ref{abstracto:logical:pi2}\\ 
      \hline
      canonical sentences &
      \ref{abstracto:canonical:general}&~\ref{abstracto:canonical:pi2}\\ 
      \hline
      algebraic interpolation & 
      \ref{abstracto:algebraic:general}&
      \ref{abstracto:algebraic:pi2}\\ 
    \end{tabular}
  \end{center}
\end{remark}
In~\cite{AU-Chen-PGP}, Chen introduces effective PGP and shows
that it entails a QCSP to CSP reduction, for the bounded alternation
QCSP. For concrete examples, such as collapsibility and switchability, he shows a QCSP to CSP
reduction even in the unbounded case~\cite[Theorem 7.11]{AU-Chen-PGP}. As a second corollary, we can
generalise this last result to effective and ``projective'' PGP, though we
formulate this in terms of sequence of adversaries.
\begin{corollary}\label{cor:EffectiveProjectivePGPImpliesQCSPInNP}
  Let $\mathcal{A}$ be a structure.  Let $\mathbf{\Omega}$ be a
  sequence of non degenerate adversaries that is effective, projective and
  polynomially bounded such that $\Omega_m$ generates $A^m$ for every
  $m\geq 1$.
  
  Let $\mathcal{A}'$ be the structure $\mathcal{A}$, possibly expanded
  with constants, at least one for each element that occurs in $\mathbf{\Omega}$.
  The problem $\QCSP(\mathcal{A})$ reduces in polynomial time to
  $\CSP(\mathcal{A}')$.
  In particular, if $\mathcal{A}$ has all constants, the problem $\QCSP_c(\mathcal{A})$ reduces in polynomial time to
  $\CSP_c(\mathcal{A})$.
\end{corollary}

\subsection{Studies of Collapsibility}
\label{sec:collapse}
Let $\mathcal{A}$ be a structure, $B\subseteq A$ and $p \geq 0$. Recall the structure
$\mathcal{A}$ is \emph{$p$-collapsible with source $B$} when for all
$m\geq 1$, for all pH-sentences $\phi$ with $m$ universal quantifiers, 
$\mathcal{A} \models \phi$ iff $\mathcal{A}\models \phi_{\restrict \Upsilon_{m,p,B}}$.
Collapsible structures are very important: to the best of our
knowledge, they are in fact the only examples of structures that enjoy a form of polynomial QCSP
to CSP reduction. This is different if one considers structures with infinitely many relations where the more
general notion of \emph{switchability} crops up~\cite{AU-Chen-PGP}.
Our abstract results of the previous section apply to both
switchability and collapsibility but we concentrate here on the latter. 
This  result applies since the underlying sequence of
adversaries are projective (see
Fact~\ref{fact:collapsible:adversary:is:projective}), as long as $p>0$
(non degenerate case).
\begin{corollary}[\textbf{In concreto}]
  \label{MainResult:InConcreto:Collapsibility}  
  Let $\mathcal{A}$ be a structure, $\emptyset \subsetneq B\subseteq A$ and $p>0$.
  The following are equivalent.
  \begin{romannum}
  \item $\mathcal{A}$ is $p$-collapsible from source $B$.
    \label{concreto:logicalcollapsibility:general}
  \item $\mathcal{A}$ is $\Pi_2$-$p$-collapsible from source $B$.
    \label{concreto:logicalcollapsibility:pi2}
  \item For every $m$, the structure $\mathcal{A}$ satisfies the
    canonical $\Pi_2$-sentence with $m\cdot |A|$ universal variables
    $\varphi_{n,\Upsilon_{m,p,B},\mathcal{A}}$.
  \item For every $m$, the structure $\mathcal{A}$ satisfies the
    canonical $\Pi_2$-sentence with $m$ universal variables
    $\varphi_{\mathscr{U},\mathcal{A}}$, where
    $\mathscr{U}=\bigcup_{\mathcal{O}\in \Upsilon_{m,p,B}} \mathcal{O}$.
  \item For every $m$, there exists a polymorphism $f$ of $\mathcal{A}$ witnessing
    that $A^m \reactivelycomposable \Upsilon_{m,p,B}$.
    \label{concreto:algebraiccollapsibility:general}
    \label{concreto:algebraiccollapsibility:iii}
  \item For every $m$, for every tuple $t$ in $A^m$, there is a
    polymorphism $f_t$ of $\mathcal{A}$ of arity $k$ at most $\binom{m}{p}.|B|$ and
    tuples $t_1,t_2,\ldots,t_k$ in $\Upsilon_{m,p,B}$ such that
    $f_t(t_1,t_2,\ldots,t_k)=t$.
    \label{concreto:algebraiccollapsibility:pi2}
  \end{romannum}
\end{corollary}
\begin{remark}\label{rem:0collapsibility}
  When $p=0$, we obtain degenerate adversaries and this is due to the fact that if a QCSP is permitted equalities, then $0$-collapsibility can never manifest (think of $\forall x, y \ x=y$).
\end{remark}
In \cite{hubie-sicomp}, Case~\ref{concreto:algebraiccollapsibility:general} of Corollary~\ref{MainResult:InConcreto:Collapsibility} is equivalent to id-Pol$(\mathcal{A})$ being $p$-collapsible (in the algebraic sense). It is proved in \cite{hubie-sicomp} that if id-Pol$(\mathcal{A})$, is $k$-collapsible (in the algebraic sense), then $\mathcal{A}$ is $k$-collapsible. We note that Corollary~\ref{MainResult:InConcreto:Collapsibility} proves the converse, finally tying together the two forms of collapsibility.

A fun application of
Corollary~\ref{MainResult:InConcreto:Collapsibility} is an alternative
proof of Proposition~\ref{prop:sc-2}. It is easy to see that a
semicomplete digraph with both a source and a sink is $1$-collapsible
with any singleton source. This is because any input sentence for
QCSP$(\mathcal{G}$), involving a universal variable $v$ in an edge
relation $E$, is false (evaluate as either the source or the sink,
depending on whether $v$ appears as the second or first entry of $E$,
respectively). The statement of the proposition now follows from
Corollary~\ref{MainResult:InConcreto:Collapsibility},
via \ref{concreto:logicalcollapsibility:general} $\Rightarrow$ \ref{concreto:algebraiccollapsibility:pi2}.

Another application of
Corollary~\ref{MainResult:InConcreto:Collapsibility} is the following (compare with
\S~\ref{sec:prPAthsWithPGP}).  
\begin{application}
\label{app:FairmontHotelTrick}
  A partially reflexive path $\mathcal{A}$ (no constants are present)
  that is quasi-loop connected has the PGP. 
\end{application}

The last two conditions of Corollary~\ref{MainResult:InConcreto:Collapsibility} provide us with
a semi-decidability result: for each $m$, we may look for a particular
polymorphism \ref{concreto:algebraiccollapsibility:general} or several
polymorphisms \ref{concreto:algebraiccollapsibility:pi2}. Instead of a sequence of
polymorphisms, we now strive for a better algebraic
characterisation. We will only be able to do so for the special case
of a singleton source, but this is the only case hitherto found in nature.  

Chen uses the following lemma to show $4$-collapsibility of bipartite
graphs and disconnected graphs~\cite[Examples 1 and
2]{Meditations}. Though, we know via a direct
argument~\cite{CiE2006} that these examples are in
fact $1$-collapsible from a singleton source. 
\begin{lemma}[Chen's lemma {\cite[Lemma
    5.13]{hubie-sicomp}}]
  \label{lem:ChensLemma}
  Let $\mathcal{A}$ be a structure with a constant $x$.
  If there is a $k$-ary polymorphism of $\mathcal{A}$ such that $f$ is surjective when restricted at any position to $\{x\}$, then $\mathcal{A}$ is $(k-1)$-collapsible from source $\{x\}$  (\mbox{i.e.} $\mathcal{A}$ has a $k$-ary Hubie polymorphism).
\end{lemma}
An interesting consequence of last section's formal work is a form of
converse of Chen's Lemma, which allows us to give an algebraic
characterisation of collapsibility from a singleton source.
\begin{proposition}
  \label{prop:singletoncollapse:WeakCharacterisation}
  Let $x$ be a constant in $\mathcal{A}$.
  The following are equivalent:
  \begin{romannum}
  \item $\mathcal{A}$ is collapsible from $\{x\}$. 
  \item $\mathcal{A}$ has a Hubie polymorphism with source $x$.
    \label{item:HubiePolx}
  \end{romannum}
\end{proposition}
In the proof of the above, for $(i)\Rightarrow (ii) \Rightarrow (i)$, we no longer control the collapsibility parameter as the
arity of our polymorphism is larger than the parameter we start with.
By inspecting more carefully the properties of the polymorphism $f$ we
get as a witness that $\mathcal{A}$ models a canonical sentence, we
may derive in fact $p$-collapsibility by an argument akin to the one
used above in the proof of Chen's Lemma. We obtain this way a nice
concrete result to counterbalance the abstract
Theorem~\ref{MainResult:InAbstracto}.
\begin{theorem}[\textbf{$p$-Collapsibility from a singleton source}]
  \label{theo:singletonSource:StrongCharacterisation}
  Let $x$ be a constant in $\mathcal{A}$ and $p>0$.
  The following are equivalent:
  \begin{romannum}
  \item $\mathcal{A}$ is $p$-collapsible from $\{x\}$.
  \item For every $m\geq 1$, the full adversary $A^m$ is reactively composable from $\Upsilon_{m,p,x}$.
  \item $\mathcal{A}$ is $\Pi_2$-$p$-collapsible from $\{x\}$. 
  \item For every $m\geq 1$, $\Upsilon_{m,p,x}$ generates $A^m$.
  \item $\mathcal{A}$ models $\phi_{n,\Upsilon_{p+1,p,x},\mathcal{A}}$
    (which implies that $\mathcal{A}$ admits a particularly well
    behaved Hubie polymorphism with source $x$ of arity $(p+1)n^p$).
  \end{romannum}
\end{theorem}
\begin{corollary}
  Given $p \geq 1$, a structure $\mathcal{A}$ and $x$ a constant in $\mathcal{A}$, we may decide
  whether $\mathcal{A}$ is $p$-collapsible from $\{x\}$.
\end{corollary}
\begin{remark}\label{rem:collapsibilityOnConservativeSourceSet}
  We say that a structure $\mathcal{A}$ is \emph{$B$-conservative} where $B$ is a subset
  of its domain iff for any polymorphism $f$ of $\mathcal{A}$ and
  any $C \subseteq B$, we have $f(C,C,\ldots,C)\subseteq C$. 
  Provided that the structure is conservative on the source set $B$,
  we may prove a similar result for $p$-Collapsibility from a conservative source.
\end{remark}


Expanding on Remark~\ref{rem:0collapsibility}, we note that if we
forbid equalities in the input to a QCSP, then we can observe the
natural case of $0$-collapsibility, to which now we turn. This is not
a significant restriction in a context of complexity, since in all but
trivial cases of a one element domain, one can propagate equality out
through renaming of variables. 

We investigated a similar notion in the context of positive equality
free first-order logic, the syntactic restriction of first-order logic
that consists of sentences using only $\exists, \forall, \land$ and $\lor$.
For this logic, relativisation of quantifiers fully explains the complexity classification of the model
checking problem (a tetrachotomy between Pspace-complete, NP-complete,
Co-NP-complete and Logspace)~\cite{LICS2011}. In
particular, a complexity in NP is characterised algebraically by the
preservation of the structure by a \emph{simple $A$-shop} (to be
defined shortly), which is equivalent to a strong form of
$0$-collapsibility since it applies not only to pH-sentences but also
to sentences of positive equality free first-order logic. 
We will show that this notion corresponds in fact to $0$-collapsibility
from a singleton source. Let us recall first some definitions.

A \emph{shop} on a set $B$, short for surjective hyper-operation, is a
function $f$ from $B$ to its powerset such that $f(x)\neq \emptyset$
for any $x$ in $B$ and for every $y$ in $B$, there exists $x$ in $B$
such that $f(x)\ni y$. An \emph{A-shop}\footnotemark{}
\footnotetext{The A does not stand for the name of the set, it is
  short for \emph{All}.} satisfies further that there is some $x$ such
that $f(x)=B$. A \emph{simple $A$-shop} satisfies further that
$|f(x')|=1$ for every $x'\neq x$.
We say that a shop $f$ is a \emph{she} of the structure $\mathcal{B}$, short for \emph{surjective
  hyper-endomorphism}, iff for any relational symbol $R$ in $\sigma$
of arity $r$, for any elements $a_1,a_2\ldots,a_r$ in $B$, if
$R(a_1,\ldots,a_r)$ holds in $\mathcal{B}$ then $R(b_1,\ldots,b_r)$
holds in $\mathcal{B}$ for any $b_1 \in f(a_1),\ldots,b_r \in f(a_r)$. 
We say that $\mathcal{B}$ \emph{admits a (simple) $A$-she} if there is a (simple) $A$-shop
$f$ that is a she of $\mathcal{B}$.

\begin{theorem}\label{theorem:0collapsibility}
  Let $\mathcal{B}$ be a finite structure.
  The following are equivalent.
  \begin{romannum}
  \item $\mathcal{B}$ is $0$-collapsible from source $\{x\}$ for some
    $x$ in $B$ for equality-free pH-sentences.\label{0collapsibility:singleton:pH}
  \item $\mathcal{B}$ admits a simple $A$-she.\label{0collapsibility:she}
  \item $\mathcal{B}$ is $0$-collapsible from source $\{x\}$ for some $x$ in $B$ for sentences of positive
    equality free first-order logic. \label{0collapsibility:mylogic}
  \end{romannum}
\end{theorem}

The above applies to singleton source only, but up to taking a power
of a structure (which satisfies the same QCSP), we may always place
ourselves in this singleton setting for $0$-collapsibility.
\begin{theorem}\label{theorem:0collapsibility:rainbowsource}
  Let $\mathcal{B}$ be a structure. The following are equivalent.
  \begin{romannum}
  \item $\mathcal{B}$ is $0$-collapsible from source $C$
  \item $\mathcal{B}^{|C|}$ is $0$-collapsible from some (any) singleton source $x$
    which is a (rainbow) $|C|$-tuple containing all elements of $C$.
  \end{romannum}
\end{theorem}


\section{Back to Complexity}
\label{sec:back-complexity}

The trichotomy of Theorem~\ref{thm:complexity} should be seen as a
companion to the following \emph{dichotomy} result.
\begin{theorem}[Theorem~1 of \cite{QCSPforests}]\label{theo:QCSPforestsDichotomyNoConstants}
  Let $\mathcal{H}$ be a p.r. path. 
  \begin{romannum}
  \item If $\mathcal{H}$ is quasi-loop-connected, then
    QCSP$(\mathcal{H})$ is in NL.
    \label{theo:QCSPforestsDichotomyNoConstants:i}
  \item Otherwise,  QCSP$(\mathcal{H})$ is Pspace-complete.
    \label{theo:QCSPforestsDichotomyNoConstants:ii}
  \end{romannum}
\end{theorem}
\noindent Case~\ref{theo:QCSPforestsDichotomyNoConstants:i} is proved
in 2 steps : a loop connected \mbox{p.r.} path is known to be in NL via a
majority polymorphism and a quasi-loop connected p.r. path is shown to
have the same QCSP via some surjective homomorphisms from powers (via
the methodology from \cite{LICS2008}). This means that we
can build a Hubie polymorphism for a quasi-loop connected p.r. path (see
Application~\ref{app:FairmontHotelTrick}). However, this polymorphism
need not be idempotent and the argument does not extend to p.r. paths
with constants. 

Using results from both of the previous sections we can now give a proof of Theorem~\ref{thm:complexity}.
\begin{proof}[Proof of Theorem~\ref{thm:complexity}]
  For Cases~\ref{thm:complexity:i} and~\ref{thm:complexity:ii}, NP
  membership follows from
  Corollary~\ref{cor:EffectiveProjectivePGPImpliesQCSPInNP} as we
  established suitable forms of PGP in Lemmas~\ref{lem:prp-1},~\ref{lem:generation} and
  \ref{lem:generation2}.
  More specifically, the Ptime membership of
  Case~\ref{thm:complexity:i} is established by the majority
  polymorphism mentioned in the proof of Lemma~\ref{lem:prp-1} (via \cite{hubie-sicomp}).
  As for Case~\ref{thm:complexity:ii}, we note in passing that
  collapsibility follows from Lemmas~\ref{lem:generation} and
  \ref{lem:generation2} which establish
  item~\ref{concreto:algebraiccollapsibility:pi2} of
  Corollary~\ref{MainResult:InConcreto:Collapsibility}. More
  importantly, NP-hardness
  follows from the classification of \cite{Pseudoforests}. 

  For Case~\ref{thm:complexity:iii}, we observe from \cite{QCSPforests} that we are Pspace-hard even without constants. 
\end{proof}
We note that the complexity classification for semicomplete digraphs from \cite{ICALP2014} is unchanged regardless of whether all constants are present (since semicompletes are cores).

\section{Conclusion}
One important application of our abstract investigation of PGP 
yields a nice characterisation in the concrete case of collapsibility,
in particular in the case of a singleton source which we now know can
be equated with preservation under a single polymorphism, namely a
Hubie polymorphism. So far, this is the only known explanation for a
complexity of a QCSP in NP which provokes the following question. 
\begin{question}
  For a structure $\mathcal{A}$, is it the case that
  QCSP($\mathcal{A}$) is in NP iff $\mathcal{A}$ admits a Hubie polymorphism?
\end{question}
In the literature, it is common to study the case of non finite
constraint languages. This means that for an infinite set of relations
over the same finite domain $\Gamma$ we study the uniform problem
QCSP($\Gamma$) which covers all problems QCSP($\mathcal{A}$) where
$\mathcal{A}$ is a structure with relations from $\Gamma$.

Typically $\Gamma$ is taken to be the invariant of some algebra.
There is an example of such a problem QCSP($\Gamma$) with a complexity in NP that is provably
not collapsible but enjoys a property similar to $p$-collapsibility, namely
\emph{$p$-switchability}~\cite{AU-Chen-PGP}, which is a special form of PGP.

For $m\geq 1$ and $\bar{\i}=(i_1,i_2,\ldots,i_{p})$ a strictly
increasing sequence in $[m-1]^{p}$, let $\mathscr{S}_{\bar{\i},p}$ be
the adversary that consists of tuples $t \in A^m$ such that each of
the following sets contain a single element: 
$\{t[j] \in A | 1 \leq j \leq i_1 \}$, $\{t[j] \in A | i_1 +1 \leq j
\leq i_2 \}$, $\ldots$, $\{t[j] \in A | i_{p} +1 \leq j \leq m\}$.
Let $\Sigma_{m,p}$ be the set of all such adversaries
$\mathscr{S}_{\bar{\i},p}$. 
Let $\mathbf{\Sigma}_p$ be the sequence of adversaries
$\bigl(\Sigma_{m,p} \bigr)_{m \in \integerset}$.

We say that a structure $A$ is \emph{$p$-switchable}
iff for every $m$ and for all pH-sentence $\phi$
with $m$ universal variable, $\mathcal{A}\models
\phi_{\restrict\Sigma_{m,p}}$ implies   $\mathscr{A} \models \phi$.

We say that a set of relations $\Gamma$ is \emph{$p$-switchable} iff
every structure $\mathcal{A}$ with relations from $\Gamma$ is $p$-switchable.

Our definition of switchability is not exactly the same as that of
Hubie Chen who uses instead a single adversary $\cup \mathscr{S}_{\bar{\i}}$
for each arity. It is a simple exercise to show that both
sequences of adversaries satisfy the hypotheses of
Theorem~\ref{MainResult:InAbstracto}.
Since the two notions are of course equivalent in
the $\Pi_2$ case via the principle of union
(Lemma~\ref{lemma:union}), they are therefore equivalent in general. 
Thus not only we can equate switchability with
its $\Pi_2$ analogue but we can also give a purely syntactic
definition of switchability as follows. A structure $\mathcal{A}$ is $p$-switchable
iff, for all $m$ and for all pH formula $\varphi$ with $m$ universal
variables $x_1,x_2,\ldots,x_m$ (in this order), $\mathcal{A} \models \varphi$ iff for all
$\bar{\i}=(i_1,i_2,\ldots,i_{p})$ a strictly increasing sequence in
$[m-1]^{p}$, $\mathcal{A}\models \varphi\land\eta_{\bar{\i}}$ where
$\eta_{\bar{\i}}$ is $\bigwedge_{0\leq \ell_1 < \ell_2 \leq
  p}\bigwedge_{i_{\ell_1}\leq j<k \leq i_{\ell_2}} x_j=x_k$. 

However, there are two limitations to our result on
switchability. Firstly, we do not have a crisp candidate for a single
polymorphism or even a sequence of polymorphisms that would endow switchability. Secondly, our findings only
hold for finite structures, where it is unclear that switchability
plays a natural role. This provokes the following question.
\begin{question}
  For every infinite set of relations $\Gamma$, is it the case that $\Gamma$ is switchable iff
  it is $\Pi_2$-switchable?
\end{question}


Going back to collapsibility, regarding the meta-question of deciding
whether a structure is collapsible, one can wonder if the parameter
$p$ of collapsibility depends on the size of the structure
$\mathcal{A}$. In particular, this would provide a positive answer to
the following.
\begin{question}
  Given a structure $\mathcal{A}$, can we decide if it is $p$-collapsible
  for some $p$?
\end{question}

A tantalising question remains.
\begin{question}
  Are there any finite algebras, minimal generating sets for whose
  powers grow sub-exponentially (\mbox{e.g.} $\Theta(2^{\sqrt{i}})$)?
\end{question}
The alternative is that finite algebras exhibit a PGP-EGP gap in general. In a sequence of three papers \emph{Growth rates of algebras}, Kearnes, Kiss and Szenderei explore this question, demonstrating all polynomial growth rates are possible.

Finally, let us return to the \emph{foundation for F\"urstenproblem} and contemplate the complexity of the QCSP. Let $\mathcal{B}$ be a finite structure. At present it is not conjectured where one might seek to prove the boundary between QCSP$(\mathcal{B})$ being in P and QCSP$(\mathcal{B})$ being NP-hard, even in the case where all constants are present. Furthermore, settling this will be at least as hard as settling the similar dichotomy for CSP. However, we would like to specifically echo the conjecture of Chen in \cite{Meditations} (where it appears written in two conjectures).

\vspace{.2cm}
\noindent \textbf{Conjecture.} Let $\mathcal{B}$ be finite and expanded with all constants; then $\QCSP_c(\mathcal{B})$ is in NP iff id-Pol$(\mathcal{B})$ has the PGP. 

\section*{Acknowledgment}
The authors would like to thank the four anonymous reviewers for their patience,
stamina and very useful suggestions, which have been a great help to
prepare the final version of this paper.

\bibliographystyle{IEEEtran}
\bibliography{local}

\newpage

\appendix


\section*{Material omitted from \S~\ref{sec:pgp-versus-egp}.}

\subsection*{Partially reflexive paths (\textsl{c.f.}\ref{sec:prpaths})}
\subsubsection*{Cases with the PGP}

\

In the proof of Lemma~\ref{lem:prp-1}, we refer to the fact that loop-connected
p.r. paths have a majority polymorphism. In the reference, it is not
fully explicit how one builds such a majority operation, and we
highlight it here for the sake of completeness.  

Let $\mathcal{P}$ be a loop-connected path labelled in ascending
natural numerical order. Let $L$ be the irreflexive component left of
the central loops and $R$ be the irreflexive component right of the
loops. If there are no loops let the whole path be in $L$. 

Recall first that the operation median over the elements of
$\mathcal{P}$ returns the argument that is neither minimal, nor
maximal when the arguments are pairwise distinct, and behave as a
majority operation otherwise.

Define
$\text{Feder}(x,y,z):=\text{median}(x,y,z)$, if all ${x,y,z}$ have the
same parity, and $\text{Feder}(x,y,z):=$max of the repeated parity,
otherwise (this operation was communicated to one of the author by email by
Tom{\'a}s Feder, hence its name).

We define
$f(x,y,z):=\text{Feder}(x,y,z)$, if ${x,y,z} \subset L$ or ${x,y,z}
\subset R$, and $f(x,y,z):=\text{median}(x,y,z)$, otherwise.
This operation $f$ is a majority polymorphism and a polymorphism of
$\mathcal{P}$.

\

\noindent \textbf{Lemma~\ref{lem:bin-pol}.}
Let $\mathcal{P}_{0^a 1^b \alpha}$, with $b>0$, be a quasi-loop-connected path on vertices $[n]$. For each $y \in [n]$ there is a binary idempotent polymorphism $f_y$ of $\mathcal{P}_{0^a 1^b \alpha}$ so that $f_y(1,x)=x$ (for all $x$) and $f_y(n,1)=y$.

\

\begin{proof}
Let $y$ be given. Suppose $\mathcal{P}_{0^a 1^b \alpha}$ is of odd length and has centre at position $q$ (the argument for even length is very similar with central vertices $q,q'$). Choose $p$ minimal ($1 \leq p\leq q$) so it is a looped vertex. Let $r$ be so that $r-q=q-p$, \mbox{i.e.} $p$ and $r$ are first and last in the block $1^b$, and we have $1 \leq p\leq q \leq r \leq n$). An idempotent binary polymorphism on domain $[n]$ may be visualised as a matrix $X$ with leading diagonal $1,\ldots,n$. We consider the top-left and bottom-left parts of the matrix $X_{tl}$ and $X_{bl}$, respectively, to include as their farthest right column the central column of the matrix $X$ at position $q$. $X_{tl}$ and $X_{bl}$ will also overlap on the bottom row of the former which is the top row of the latter. Let us consider what constraints a polymorphism must satisfy. Across the whole matrix, diagonal neighbours must be adjacent elements. In $X_{tl}$, in fact, only the diagonals are needed to be considered to satisfy polymorphism. But in $X_{bl}$ (indeed the whole bottom half) there may be some horizontal lines that must satisfy the adjacency condition and in the right half there might be some vertical lines that need to satisfy this adjacency condition too. To see an example of this we direct the reader to $\mathcal{P}_{00001110110}$ in Figure~\ref{fig:third}. We will rebuild $X_{tl}$ and $X_{bl}$ to satisfy all horizontals, even though we do not need them all, and the right half of the matrix will satisfy all potential vertical lines.

When viewed as a matrix, the entire right half (from and including the middle column $q$) will obey $f(u,v)=v$. We now turn our attention to $X_{tl}$. The farthest right column of $X_{tl}$ is already set to $q$, and we will set the entire bottom row to $q$. We now remove these already-set positions and then set the farthest right column and bottom row of the remainder of $X_{tl}$ to $q-1$. We iterate this until we reach and have done this for $p$. Note that this is consistent with idempotency. We have now filled in $X$ other than a $(p-1) \times (p-1)$ matrix in the top-left which we call  $X'_{tl}$ and a  $(q-1) \times (q-1)$ matrix in the bottom-left which we still call  $X_{bl}$. This is depicted in Figure~\ref{fig:first} and satisfies all local conditions for polymorphism. The matrix $X'_{tl}$ must additionally satisfy leading diagonal idempotency and must also satisfy the boundary condition of $p$ against its right-most column and bottom row. The matrix $X_{bl}$ must satisfy position $(n,1)$ being $y$ and the boundary condition of $q$ against its right-most column and top row.

(Construction of $X'_{tl}$.) We explain how to fill in position $(1,i)$ and $(i,1)$ for $i \in [p-1]$ because each diagonal proceeding towards the centre of the matrix will contain an increasing arithmetic sequence with step $1$. Set $(1,i)$ to be $i$ and $(i,1)$ to be $i$ (when $i$ is odd) and $i+1$ (when $i$ is even). A simple calculation now yields the precise specification: if $\lambda < \mu$, set $(\lambda,\mu)$ to $\mu$;  if $\lambda > \mu$, set $(\lambda,\mu)$ to $\lambda$ (if $\lambda-\mu+1$ is odd) and to $\lambda+1$ (if $\lambda-\mu+1$ is even). It is easy to see that this satisfies polymorphism. Indeed, it satisfies polymorphism on the horizontals where it is not necessary (but will become necessary for $X_{bl}$).

(Construction of $X_{bl}$.) The upward diagonal from $y$ at position $(n,1)$ to $q$ is filled $y,y \pm 1 \ldots,q,\ldots,q$. That is, if $y\leq q$ we increase by one until we reach $q$ and then repeat $q$, and if  $y\geq q$ we decrease by one until we reach $q$ and then repeat $q$.

All rows and columns in $X_{bl}$ that contain a vertex $z \in \{p,\ldots,q,\ldots,r\}$ on the upward diagonal from $(n,1)$ are now filled in with $z$. At this point we are left with some $s \times s$ submatrix $X'_{bl}$ of $X_{bl}$ not filled in. $X'_{bl}$ might be empty if  $y \in \{p,\ldots,q,\ldots,r\}$, but if $X'_{bl}$ is not empty then we have the boundary condition of either $p$ or $r$ against its right-most column and top row. We now fill this in in precisely the dual fashion to our filling in of $X'_{tl}$. We will give the argument when the boundary condition is $r$ (the other case of boundary $p$ being very similar). We explain how to fill in position $(n,i)$ and $(i,n)$ for $i \in \{n,\ldots,n-s+1)$ because each diagonal proceeding towards the centre of the matrix will contain a decreasing (increasing if boundary is instead $p$) arithmetic sequence with step $1$. Set $(n,i)$ to be $y-i+1$ and $(i,n)$ to be $i-n+y$ (when $i$ is odd) and $i-1-n+y$ (when $i$ is even). It is not hard to see that this satisfies polymorphism, even on its horizontals.

Two examples, for the graph $\mathcal{P}_{0^41^3\alpha}$ with $|\alpha|=4$,   are given in Figure~\ref{fig:second}. The left-hand example is for ($n=11$ where $p=5, q=6,r=7$ and) $y=10$; and the right-hand example is for  ($n=11$ where $p=5, q=6,r=7$ and) $y=3$.
\end{proof}
\begin{figure}
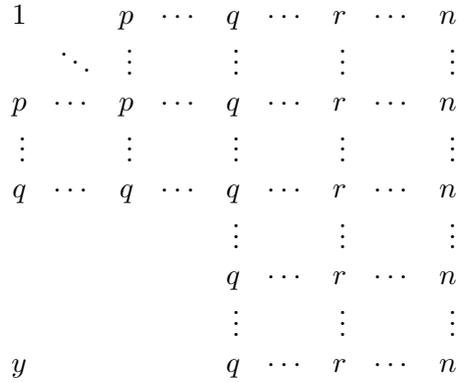

\centering
\[
\begin{array}{rrrrrrrrr}
1 &  & p & \cdots & q & \cdots & r & \cdots & n \\
 & \ddots & \vdots & & \vdots & & \vdots & & \vdots \\
p & \cdots & p & \cdots & q & \cdots & r & \cdots & n \\
\vdots &  & \vdots &  & \vdots  &  & \vdots &  & \vdots \\
q & \cdots & q & \cdots & q  & \cdots & r & \cdots & n \\
 &  &  &  & \vdots  &  & \vdots &  & \vdots \\
 &  &  &  & q & \cdots & r & \cdots & n \\
 &  &  & & \vdots &  & \vdots &  & \vdots \\
y &  &  &  & q & \cdots & r & \cdots & n \\
\end{array}
\]
\caption{First part of the construction for the proof of Lemma~\ref{lem:bin-pol}.}
\label{fig:first}
\end{figure}
\begin{figure}
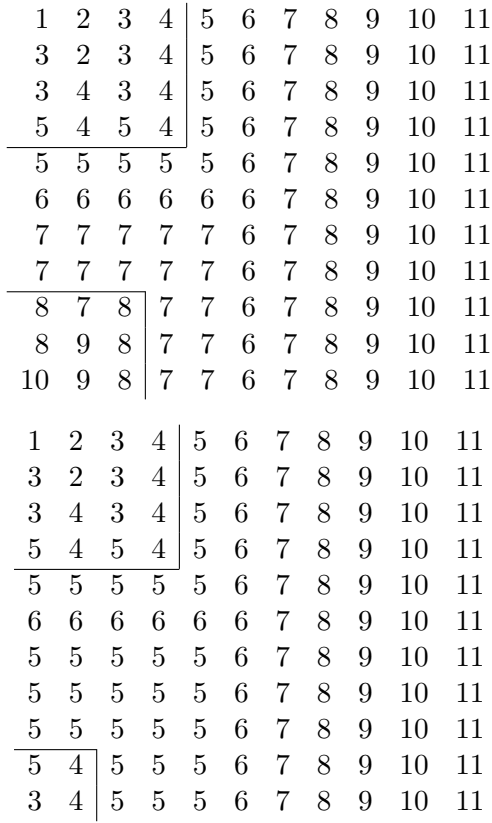

\centering
$
\begin{array}{rrr r rrrrrrr}
1 & 2 & 3 & \multicolumn{1}{r|}{4}& 5 & 6 & 7 & 8 & 9 & 10 & 11 \\
3 & 2 & 3 & \multicolumn{1}{r|}{4}& 5 & 6 & 7 & 8 & 9 & 10 & 11 \\
3 & 4 & 3 & \multicolumn{1}{r|}{4} & 5 & 6 & 7 & 8 & 9 & 10 & 11 \\
5 & 4 & 5 & \multicolumn{1}{r|}{4} & 5 & 6 & 7 & 8 & 9 & 10 & 11 \\
\cline{1-4}
5 & 5 & 5 & 5 & 5 & 6 & 7 & 8 & 9 & 10 & 11 \\
6 & 6 & 6 & 6 & 6 & 6 & 7 & 8 & 9 & 10 & 11 \\
7 & 7 & 7 & 7 & 7 & 6 & 7 & 8 & 9 & 10 & 11 \\
7 & 7 & 7 & 7 & 7 & 6 & 7 & 8 & 9 & 10 & 11 \\
\cline{1-3}
8 & 7 & \multicolumn{1}{r|}{8} & 7 & 7 & 6 & 7 & 8 & 9 & 10 & 11 \\
8 & 9 & \multicolumn{1}{r|}{8} & 7 & 7 & 6 & 7 & 8 & 9 & 10 & 11 \\
10 & 9 & \multicolumn{1}{r|}{8} & 7 & 7 & 6 & 7 & 8 & 9 & 10 & 11 \\
\end{array}
$

\vspace{0.3cm}

$
\begin{array}{rr rr rrrrrrr}
1 & 2 & 3 & \multicolumn{1}{r|}{4} & 5 & 6 & 7 & 8 & 9 & 10 & 11 \\
3 & 2 & 3 & \multicolumn{1}{r|}{4} & 5 & 6 & 7 & 8 & 9 & 10 & 11 \\
3 & 4 & 3 & \multicolumn{1}{r|}{4} & 5 & 6 & 7 & 8 & 9 & 10 & 11 \\
5 & 4 & 5 & \multicolumn{1}{r|}{4} & 5 & 6 & 7 & 8 & 9 & 10 & 11 \\
\cline{1-4} 
5 & 5 & 5 & 5 & 5 & 6 & 7 & 8 & 9 & 10 & 11 \\
6 & 6 & 6 & 6 & 6 & 6 & 7 & 8 & 9 & 10 & 11 \\
5 & 5 & 5 & 5 & 5 & 6 & 7 & 8 & 9 & 10 & 11 \\
5 & 5 & 5 & 5 & 5 & 6 & 7 & 8 & 9 & 10 & 11 \\
5 & 5 & 5 & 5 & 5 & 6 & 7 & 8 & 9 & 10 & 11 \\
\cline{1-2}
5 & \multicolumn{1}{r|}{4} & 5 & 5 & 5 & 6 & 7 & 8 & 9 & 10 & 11 \\
3 & \multicolumn{1}{r|}{4} & 5 & 5 & 5 & 6 & 7 & 8 & 9 & 10 & 11 \\
\end{array}
$
\caption{Two polymorphism of the graph $\mathcal{P}_{0^41^3\alpha}$, with $\alpha$ any string of $0$s and $1$s of length $4$. The lines indicate the boundaries of $X'_{tl}$ and $X'_{bl}$.} 
\label{fig:second}
\end{figure}
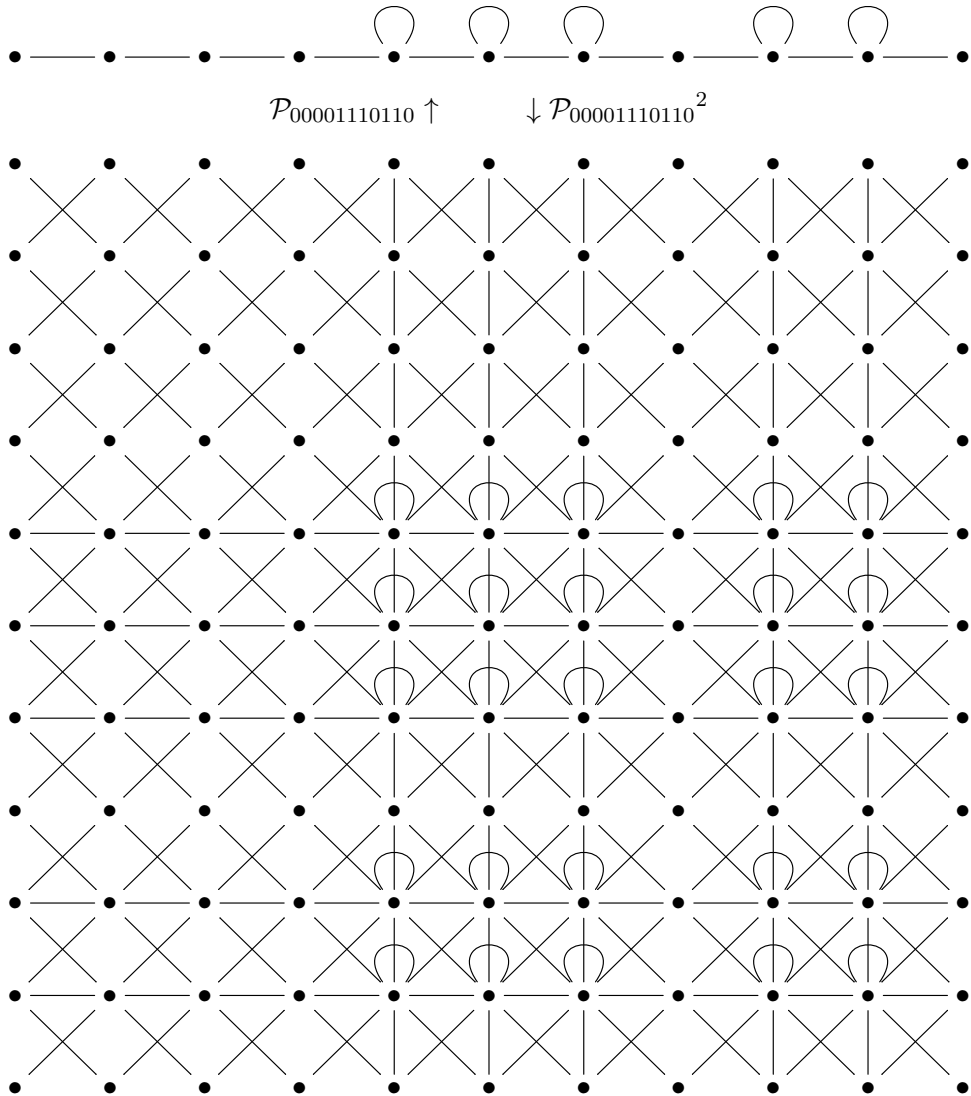
\begin{figure*}
  \[
  \xymatrix{
    \bullet \ar@{-}[r] &\bullet \ar@{-}[r] & \bullet \ar@{-}[r] & \bullet \ar@{-}[r] & \bullet \ar@{-}[r]  \ar@{-}@(ul,ur)[] & \bullet \ar@{-}[r]  \ar@{-}@(ul,ur)[] & \bullet \ar@{-}[r]  \ar@{-}@(ul,ur)[] & \bullet \ar@{-}[r] & \bullet \ar@{-}[r]  \ar@{-}@(ul,ur)[] & \bullet \ar@{-}[r]  \ar@{-}@(ul,ur)[] & \bullet \\
  }
  \]
  \[\mathcal{P}_{00001110110} \uparrow \ \ \ \ \ \ \ \ \ \downarrow {\mathcal{P}_{00001110110}}^2 \]
  \[
  \xymatrix{
    \bullet  \ar@{-}[dr] &\bullet \ar@{-}[dl] \ar@{-}[dr] & \bullet \ar@{-}[dl] \ar@{-}[dr] & \bullet \ar@{-}[dl] \ar@{-}[dr] & \bullet \ar@{-}[dl] \ar@{-}[dr] \ar@{-}[d] & \bullet  \ar@{-}[dl] \ar@{-}[dr]  \ar@{-}[d] & \bullet  \ar@{-}[dl] \ar@{-}[dr]  \ar@{-}[d] & \bullet \ar@{-}[dl] \ar@{-}[dr] & \bullet \ar@{-}[dl] \ar@{-}[dr]  \ar@{-}[d] & \bullet \ar@{-}[dl] \ar@{-}[dr]  \ar@{-}[d] & \bullet \ar@{-}[dl] \\
    \bullet  \ar@{-}[dr] &\bullet \ar@{-}[dl] \ar@{-}[dr] & \bullet \ar@{-}[dl] \ar@{-}[dr] & \bullet \ar@{-}[dl] \ar@{-}[dr] & \bullet \ar@{-}[dl] \ar@{-}[dr] \ar@{-}[d] & \bullet  \ar@{-}[dl] \ar@{-}[dr]  \ar@{-}[d] & \bullet  \ar@{-}[dl] \ar@{-}[dr]  \ar@{-}[d] & \bullet \ar@{-}[dl] \ar@{-}[dr] & \bullet \ar@{-}[dl] \ar@{-}[dr]  \ar@{-}[d] & \bullet \ar@{-}[dl] \ar@{-}[dr]  \ar@{-}[d] & \bullet \ar@{-}[dl] \\
    \bullet  \ar@{-}[dr] &\bullet \ar@{-}[dl] \ar@{-}[dr] & \bullet \ar@{-}[dl] \ar@{-}[dr] & \bullet \ar@{-}[dl] \ar@{-}[dr] & \bullet \ar@{-}[dl] \ar@{-}[dr] \ar@{-}[d] & \bullet  \ar@{-}[dl] \ar@{-}[dr]  \ar@{-}[d] & \bullet  \ar@{-}[dl] \ar@{-}[dr]  \ar@{-}[d] & \bullet \ar@{-}[dl] \ar@{-}[dr] & \bullet \ar@{-}[dl] \ar@{-}[dr]  \ar@{-}[d] & \bullet \ar@{-}[dl] \ar@{-}[dr]  \ar@{-}[d] & \bullet \ar@{-}[dl] \\
    \bullet  \ar@{-}[dr] &\bullet \ar@{-}[dl] \ar@{-}[dr] & \bullet \ar@{-}[dl] \ar@{-}[dr] & \bullet \ar@{-}[dl] \ar@{-}[dr] & \bullet \ar@{-}[dl] \ar@{-}[dr] \ar@{-}[d] & \bullet  \ar@{-}[dl] \ar@{-}[dr]  \ar@{-}[d] & \bullet  \ar@{-}[dl] \ar@{-}[dr]  \ar@{-}[d] & \bullet \ar@{-}[dl] \ar@{-}[dr] & \bullet \ar@{-}[dl] \ar@{-}[dr]  \ar@{-}[d] & \bullet \ar@{-}[dl] \ar@{-}[dr]  \ar@{-}[d] & \bullet \ar@{-}[dl] \\
    \bullet \ar@{-}[r]  \ar@{-}[dr] &\bullet \ar@{-}[r]  \ar@{-}[dl] \ar@{-}[dr] & \bullet \ar@{-}[r]  \ar@{-}[dl] \ar@{-}[dr] & \bullet \ar@{-}[r] \ar@{-}[dl] \ar@{-}[dr] & \bullet \ar@{-}[r] \ar@{-}[dl]  \ar@{-}[d] \ar@{-}[dr] \ar@{-}@(ul,ur)[] & \bullet \ar@{-}[r] \ar@{-}[dl] \ar@{-}[dr]  \ar@{-}[d] \ar@{-}@(ul,ur)[] & \bullet \ar@{-}[r] \ar@{-}[dl] \ar@{-}[dr]  \ar@{-}[d] \ar@{-}@(ul,ur)[] & \bullet \ar@{-}[r] \ar@{-}[dl] \ar@{-}[dr] & \bullet \ar@{-}[r] \ar@{-}[dl] \ar@{-}[dr]  \ar@{-}[d] \ar@{-}@(ul,ur)[] & \bullet \ar@{-}[r] \ar@{-}[dl] \ar@{-}[dr]  \ar@{-}[d] \ar@{-}@(ul,ur)[] & \bullet \ar@{-}[dl] \\
    \bullet \ar@{-}[r]  \ar@{-}[dr] &\bullet \ar@{-}[r]  \ar@{-}[dl] \ar@{-}[dr] & \bullet \ar@{-}[r]  \ar@{-}[dl] \ar@{-}[dr] & \bullet \ar@{-}[r] \ar@{-}[dl] \ar@{-}[dr] & \bullet \ar@{-}[r] \ar@{-}[dl]  \ar@{-}[d] \ar@{-}[dr] \ar@{-}@(ul,ur)[] & \bullet \ar@{-}[r] \ar@{-}[dl] \ar@{-}[dr]  \ar@{-}[d] \ar@{-}@(ul,ur)[] & \bullet \ar@{-}[r] \ar@{-}[dl] \ar@{-}[dr]  \ar@{-}[d] \ar@{-}@(ul,ur)[] & \bullet \ar@{-}[r] \ar@{-}[dl] \ar@{-}[dr] & \bullet \ar@{-}[r] \ar@{-}[dl] \ar@{-}[dr]  \ar@{-}[d] \ar@{-}@(ul,ur)[] & \bullet \ar@{-}[r] \ar@{-}[dl] \ar@{-}[dr]  \ar@{-}[d] \ar@{-}@(ul,ur)[] & \bullet \ar@{-}[dl] \\
    \bullet \ar@{-}[r]  \ar@{-}[dr] &\bullet \ar@{-}[r]  \ar@{-}[dl] \ar@{-}[dr] & \bullet \ar@{-}[r]  \ar@{-}[dl] \ar@{-}[dr] & \bullet \ar@{-}[r] \ar@{-}[dl] \ar@{-}[dr] & \bullet \ar@{-}[r] \ar@{-}[dl]  \ar@{-}[d] \ar@{-}[dr] \ar@{-}@(ul,ur)[] & \bullet \ar@{-}[r] \ar@{-}[dl] \ar@{-}[dr]  \ar@{-}[d] \ar@{-}@(ul,ur)[] & \bullet \ar@{-}[r] \ar@{-}[dl] \ar@{-}[dr]  \ar@{-}[d] \ar@{-}@(ul,ur)[] & \bullet \ar@{-}[r] \ar@{-}[dl] \ar@{-}[dr] & \bullet \ar@{-}[r] \ar@{-}[dl] \ar@{-}[dr]  \ar@{-}[d] \ar@{-}@(ul,ur)[] & \bullet \ar@{-}[r] \ar@{-}[dl] \ar@{-}[dr]  \ar@{-}[d] \ar@{-}@(ul,ur)[] & \bullet \ar@{-}[dl] \\
    \bullet  \ar@{-}[dr] &\bullet \ar@{-}[dl] \ar@{-}[dr] & \bullet \ar@{-}[dl] \ar@{-}[dr] & \bullet \ar@{-}[dl] \ar@{-}[dr] & \bullet \ar@{-}[dl] \ar@{-}[dr] \ar@{-}[d] & \bullet  \ar@{-}[dl] \ar@{-}[dr]  \ar@{-}[d] & \bullet  \ar@{-}[dl] \ar@{-}[dr]  \ar@{-}[d] & \bullet \ar@{-}[dl] \ar@{-}[dr] & \bullet \ar@{-}[dl] \ar@{-}[dr]  \ar@{-}[d] & \bullet \ar@{-}[dl] \ar@{-}[dr]  \ar@{-}[d] & \bullet \ar@{-}[dl] \\
    \bullet \ar@{-}[r]  \ar@{-}[dr] &\bullet \ar@{-}[r]  \ar@{-}[dl] \ar@{-}[dr] & \bullet \ar@{-}[r]  \ar@{-}[dl] \ar@{-}[dr] & \bullet \ar@{-}[r] \ar@{-}[dl] \ar@{-}[dr] & \bullet \ar@{-}[r] \ar@{-}[dl]  \ar@{-}[d] \ar@{-}[dr] \ar@{-}@(ul,ur)[] & \bullet \ar@{-}[r] \ar@{-}[dl] \ar@{-}[dr]  \ar@{-}[d] \ar@{-}@(ul,ur)[] & \bullet \ar@{-}[r] \ar@{-}[dl] \ar@{-}[dr]  \ar@{-}[d] \ar@{-}@(ul,ur)[] & \bullet \ar@{-}[r] \ar@{-}[dl] \ar@{-}[dr] & \bullet \ar@{-}[r] \ar@{-}[dl] \ar@{-}[dr]  \ar@{-}[d] \ar@{-}@(ul,ur)[] & \bullet \ar@{-}[r] \ar@{-}[dl] \ar@{-}[dr]  \ar@{-}[d] \ar@{-}@(ul,ur)[] & \bullet \ar@{-}[dl] \\
    \bullet \ar@{-}[r]  \ar@{-}[dr] &\bullet \ar@{-}[r]  \ar@{-}[dl] \ar@{-}[dr] & \bullet \ar@{-}[r]  \ar@{-}[dl] \ar@{-}[dr] & \bullet \ar@{-}[r] \ar@{-}[dl] \ar@{-}[dr] & \bullet \ar@{-}[r] \ar@{-}[dl]  \ar@{-}[d] \ar@{-}[dr] \ar@{-}@(ul,ur)[] & \bullet \ar@{-}[r] \ar@{-}[dl] \ar@{-}[dr]  \ar@{-}[d] \ar@{-}@(ul,ur)[] & \bullet \ar@{-}[r] \ar@{-}[dl] \ar@{-}[dr]  \ar@{-}[d] \ar@{-}@(ul,ur)[] & \bullet \ar@{-}[r] \ar@{-}[dl] \ar@{-}[dr] & \bullet \ar@{-}[r] \ar@{-}[dl] \ar@{-}[dr]  \ar@{-}[d] \ar@{-}@(ul,ur)[] & \bullet \ar@{-}[r] \ar@{-}[dl] \ar@{-}[dr]  \ar@{-}[d] \ar@{-}@(ul,ur)[] & \bullet \ar@{-}[dl] \\
    \bullet & \bullet & \bullet & \bullet & \bullet & \bullet & \bullet & \bullet & \bullet & \bullet & \bullet \\
  }
  \]
  \caption{Path $\mathcal{P}_{0^41^301^20}$ and its square} 
  \label{fig:third}
\end{figure*}

\
\begin{lemma}
  \label{lem:bin-pol2}
  Let $\mathcal{P}_{0^a\alpha}$ be a quasi-loop-connected path on vertices $[n]$ (that is not of the form $\mathcal{P}_{0^a 1^b \alpha}$ with $|\alpha|=a$). For each $y \in [n]$ there is a binary idempotent polymorphism $f_y$ of $\mathcal{P}_{0^a\alpha}$ so that $f_y(1,x)=x$ (for all $x$) and either $f_y(n,1)=y$ or $f_y(n,2)=y$.
\end{lemma}
%

\begin{proof}
Suppose first that $y \in \{1,2\}$. Assume $n$ is even (the argument for the odd case is very similar). Our proof has similarities to that of Lemma~\ref{lem:bin-pol}. We will rebuild $X_{tl}$ roughly as before, but now we rebuild  $X_{bl}$ as a mirror image of  $X_{tl}$. We will set the columns $n/2+1,\ldots,n$ of our matrix to be full columns of $n/2+1,\ldots,n$, respectively. Take the remainder of the matrix, on columns $1,\ldots,n/2$ and split it into two across its central horizontal. Call the top left $X_{tl}$ and the bottom left $X_{bl}$. We will fill in $X_{tl}$ according to the construction of $X'_{tl}$ from Lemma~\ref{lem:bin-pol}. We will now fill $X_{bl}$ to be a mirror image of $X_{tl}$ in the central horizontal. An example of this is depicted in Figure~\ref{fig:second-2}.

Now for general $y \in \{z,z+1\}$ with $z$ odd, we shift the horizontal split between $X_{tl}$ and $X_{bl}$ downwards (making $X_{tl}$ larger). The split will be just after row $n/2+(z-1)/2$, \mbox{i.e.} at $z=n-1$ the matrix $X_{bl}$ is empty. We now build $X_{tl}$ according to the construction of $X'_{tl}$ from Lemma~\ref{lem:bin-pol} and now fill $X_{bl}$ to be a mirror image of the bottom part  of $X_{tl}$ in the horizontal just after row $n/2+(z-1)/2$. An example of this is depicted in Figure~\ref{fig:second-2}.
\end{proof}
\begin{figure}
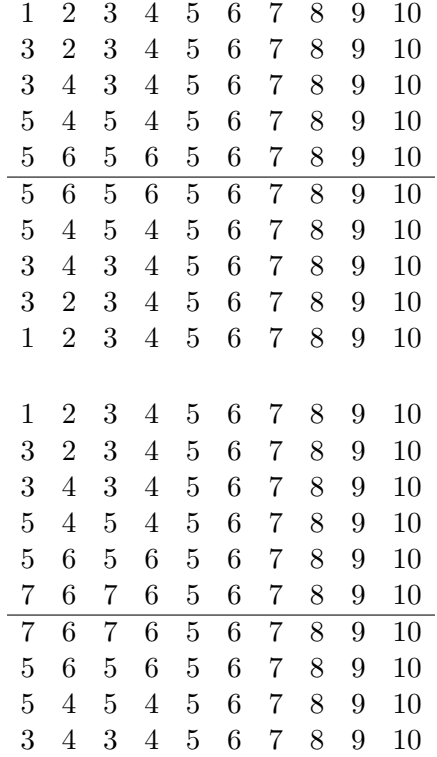

\centering
$
\begin{array}{rrrrrrrrrr}
1 & 2 & 3 & 4 & 5 & 6 & 7 & 8 & 9 & 10 \\
3 & 2 & 3 & 4 & 5 & 6 & 7 & 8 & 9 & 10 \\
3 & 4 & 3 & 4 & 5 & 6 & 7 & 8 & 9 & 10 \\
5 & 4 & 5 & 4 & 5 & 6 & 7 & 8 & 9 & 10 \\
5 & 6 & 5 & 6 & 5 & 6 & 7 & 8 & 9 & 10 \\
\hline
5 & 6 & 5 & 6 & 5 & 6 & 7 & 8 & 9 & 10 \\
5 & 4 & 5 & 4 & 5 & 6 & 7 & 8 & 9 & 10 \\
3 & 4 & 3 & 4 & 5 & 6 & 7 & 8 & 9 & 10 \\
3 & 2 & 3 & 4 & 5 & 6 & 7 & 8 & 9 & 10 \\
1 & 2 & 3 & 4 & 5 & 6 & 7 & 8 & 9 & 10 \\
\end{array}
$

\vspace{0.5cm}

$
\begin{array}{rrrrrrrrrrr}
1 & 2 & 3 & 4 & 5 & 6 & 7 & 8 & 9 & 10 \\
3 & 2 & 3 & 4 & 5 & 6 & 7 & 8 & 9 & 10 \\
3 & 4 & 3 & 4 & 5 & 6 & 7 & 8 & 9 & 10 \\
5 & 4 & 5 & 4 & 5 & 6 & 7 & 8 & 9 & 10 \\
5 & 6 & 5 & 6 & 5 & 6 & 7 & 8 & 9 & 10 \\
7 & 6 & 7 & 6 & 5 & 6 & 7 & 8 & 9 & 10 \\
\hline
7 & 6 & 7 & 6 & 5 & 6 & 7 & 8 & 9 & 10 \\
5 & 6 & 5 & 6 & 5 & 6 & 7 & 8 & 9 & 10 \\
5 & 4 & 5 & 4 & 5 & 6 & 7 & 8 & 9 & 10 \\
3 & 4 & 3 & 4 & 5 & 6 & 7 & 8 & 9 & 10 \\
\end{array}
$
\caption{Examples for the proof of Lemma~\ref{lem:bin-pol2}. The line indicates the axis of symmetry for the mirror image.}
\label{fig:second-2}
\end{figure}

\noindent \textbf{Lemma~\ref{lem:generation2}.}
Let $\mathcal{P}_{0^a \alpha}$, for $|\alpha|\in \{a,a-1\}$, be a quasi-loop-connected path on vertices $[n]$  (that is not of the form $\mathcal{P}_{0^a 1^b \alpha}$ with $|\alpha|=a$). Let $\mathbb{A}$ be the algebra specified by $\mathrm{id\mbox{-}Pol}(\mathcal{P}_{0^a \alpha})$. For each $m$, $\mathbb{A}^m$ is generated from the $2n+2$ $m$-tuples $(1,1,\ldots,1),$ $(2,2,\ldots,2),(n,1,\ldots,1),$ $(1,n,\ldots,1),$ $\ldots, (1,1,\ldots,n)$,$(n,2,\ldots,2), (2,n,\ldots,2), \ldots, (2,2,\ldots,n)$.

\

\begin{proof}
 The proof is as in the Lemma~\ref{lem:generation} but relies upon Lemma~\ref{lem:bin-pol2} in place of Lemma~\ref{lem:bin-pol}.
\end{proof}

\subsubsection*{Cases with the EGP}

\


For a digraph $\mathcal{H}$, the distance, $d_{\mathcal{H}}$,  between two vertices is the number of edges in a shortest path connecting them. By  $\mathcal{H}^n$ we mean  the tensor product  of $\mathcal{H}$ with itself $n$ times. 

We note that polymorphisms do not increase distances in graphs, i.e. if $f$ is an $n$-ary polymorphism of $\mathcal{H}$  and $u, v\in H^n$ then 
$d_{\mathcal{H}^n}(u, v) \le d_{\mathcal{H}}(f(u), f(v))$.

Lemma~\ref{lem:cati} is proved by induction on the arity of the
polymorphisms. We deal first with the base case.
\begin{lemma}
  \label{generalbinary}
  Let $\alpha$ be any sequence of zeros and ones. All idempotent
  binary polymorphisms of $\mathcal{P}_{10\alpha01}$ are projections.
\end{lemma} 
\begin{proof}
We label the vertices of $\algP=\mathcal{P}_{10\alpha 01}$ left to right over $0, 1, \ldots, t=|\algP|-1$
and start by showing that any binary polymorphism $f$ of $\algP$ 
must satisfy the following 
$$f(i, j)\le \max \{i, j\} \ \ {\rm and } \ \  f(i, j)\ge \min\{i, j\},$$
with $i, j=0, \ldots, t$ and considering the natural  linear ordering of the labelling of vertices of $\algP$.
Assume, for a contradiction, that there exist $i,j$  such that $f(i, j) =k$ with $k>i, j $. Without loss of generality we assume that $i<j$. There exists a path  of length at most $j$ from $f(i, j)$ to $f(0, 0)=0$, via the vertices $f(i-1, j-1), \ldots, f(0, 1), f(0, 0)$, but clearly $d_{\mathcal{H}}(0, k)=k$, so we get a contradiction. Dually, we can show that we also cannot have $k<i, j$.

We now show that $f_{|\{x, x+k\}}$ is, without loss of generality,  the first projection, by induction on $k\ge 1$.
 
There is an edge, in $\mathcal{P}$, from $f(0,1)$ and from $f(1,0)$ to $f(0,  0)=0$, so $f(0,  1), f(1,0)\in \{1,0\}$. There is also an edge from $f(0, 1)$ to $f(1,0)$, so they cannot both be equal to $1$. 
In a similar way we can check that $f(t, t-1), f(t-1, t)\in  \{t-1, t\}$ and  they cannot both be equal to $t-1$.

We have  $d_{\mathcal{H}}(f(0,1), f(t-1, t), d_{\mathcal{H}}(f(1,0), f(t, t-1)\le t-1$,  since $f(t, t-1)$ and $ f(t-1,t)$ cannot both be equal to $t-1$, this immediately implies 
that $f(1,0)$ and $f(0,1)$ cannot both be equal to $0$. Hence it follows that $f_{|\{1,2\}}$  
must be a projection.  Assume, without loss of generality, that it is the first projection.

To be able to get the correct distances from $f(0,1)$ to $f(t-1,t)$ we must have that $f$ restricted 
to any two consecutive vertices must be the first projection, i.e. $f(x, x+1)=x$ and $f(x+1, x)=x+1$ for all $x=0, \ldots, t-1$.

Now, assume that $f_{|\{x, x+l\}}$ is the first projection, for all $l<m$ and all $x=0, \ldots, t-l$. We show that $f_{|\{x, x+m\}}$ is also the first projection, by induction on $x$.
For the base case $x=0$, we know that  there is an edge from $f(0, m)$ to $f(0, m-1)$ and an edge from $f(m,0)$ to $f(m-1, 0)$.
 By the inductive hypothesis, $f(0, m-1)=0$ and $f(m-1, 0)=m-1$, so we must have $f(0, m)\in \{1,0\}$ and $f(m,0)\in \{m-2,m-1,m\}$.

Also, there is an edge from $f(0,m)$ to $f(1,m-1)$,  by the inductive hypothesis $f(1,m-1)=1$, so we must have $f(0,m)=0$. 
We now just need to consider the case $f(m,0)$.

{\it Case 1:} Suppose that $f(m, 0)=m-2$;   $d_{\mathcal{H}}(f(m, 0), f(t, t-(m+1))\le t-m$, and by the inductive hypothesis  $f(t, t-(m+1))=t$. So $d_{\mathcal{H}}(f(m,0), t)\le t- m$, but $d_{\mathcal{H}}(m-2, t)=t-(m-2)$, so we get a contradiction.

{\it Case 2:} Suppose that $f(m, 0)=m-1$. Since there is an edge from $f(m, 0)$ to $f(m-1, 0)$, and $f(m-1, 0)=m-1$ by the inductive hypothesis, it follows that $m-1$ must be a loop. Now, there is an edge from $f(m,0)$ to $f(m+1,1)$ and from this to $f(m,2)$. Since $f(m,2)=m$, by the inductive hypothesis. We have that $f(m+1,1)\in \{m-1,m\}$. If $f(m+1,1)=m-1$ we get a similar contradiction as in Case 1, so we must have $f(m+1,1)=m$, which also implies that $m$ must be a loop. We now move on to $f(m+2,2)$ and using the same reasoning we get that $m+2$ must be a loop. Carrying on in this way we will eventually reach a contradiction since the vertex $t-1$ does not have a loop; unless $m=t$, in which case $f(t, 0)$ is a loop and we immediately have $f(t, 0)=t$. 

Hence we must have $f(m,0)=m$. This proves the base case.

Assume now that $f_{|\{x, x+m\}}$ is the first projection for all $x<b$. We show that $f_{|\{b, b+m\}}$ is also the first projection.  There are  edges from $f(b, b+m)$ and $f(b+m, b)$ to $f(b-1, b-1+m)$ and $f(b-1+m, b-1)$ respectively. By the inductive hypothesis, $f(b-1, b-1+m)=b-1$ and $f(b-1+m, b-1)=b-1+m$.
So we have $f(b, b+m)\in \{b-2,b-1,b\}$ and $f(b+m, b)\in \{b-2+m, b-1+m, b+m\}$.

There is  an  edge from  $f(b, b+m)$  to $f(b+1, b+m-1)$, and,  by the inductive hypothesis, $f(b+1, b+m-1)=b+1$. So we must have $f(b, b+m)\in \{b, b+1, b+2\}$, it immediately follows that $f(b, b+m)=b$.
Like above, in Cases 1 and 2, we can show that we also must have $f(b+m, b)=b+m$.

This proves the lemma. 
\end{proof}

\

\noindent \textbf{Lemma~\ref{lem:cati}.}
Let $\alpha$ be any sequence of zeros and ones. All idempotent polymorphisms of $\mathcal{P}_{10\alpha01}$ are projections.

\begin{proof}
Let $\mathcal{P}=\mathcal{P}_{10\alpha01}$ and label the vertices of $\mathcal{P}$ over $[t]$ with $ t=|\mathcal{P}|$ left to right. Let $n\ge 2$ be arbitrary and let $f(x_1, \ldots, x_n)$ be any idempotent $n$-ary polymorphism of $\mathcal{P}$. We prove the lemma  by induction on $n$, with base case given by Lemma~\ref{generalbinary}.

Assume now that the lemma holds for any $n<k$, i.e. we have  $f(x_1, \ldots, x_n)=x_1$ for any $x_1, \ldots, x_n$ vertices of $\algP$ and any $n<k$. Let us consider the case when $f$ is a polymorphism of arity $k$. We will show that $f(x_1, \ldots, x_k)$ is also the first projection.

{\it Case 1:} $x_1$ is not the left-most nor the right-most  element of $x_1, \ldots, x_k$.

In this case we know that $d(f(x_1, \ldots, x_k), f(1, y_2, \ldots, y_k))\le x_1-1$, where $y_i=x_i-x_1$  if $x_i>x_1$ and it is $1$ otherwise. Now at least one of the $y_i$s equals to $1$, so at this stage $f(1, y_2, \ldots, y_k)$ matches  a polymorphism of arity smaller than $k$ and we can apply the inductive hypothesis, so that  $f(1, y_2, \ldots, y_k)=1$.   It follows that $d(f(x_1, \ldots, x_k), 1)\le x_1-1$.
In a similar way we obtain that $d(f(x_1, \ldots, x_k), f(t, z_2, \ldots, z_k))\le t -x_1-1$, so that $d(f(x_1, \ldots, x_k), t)\le  t -x_1-1$. It follows that $f(x_1, \ldots, x_k)=x_1$.

{\it Case 2:} $x_1=1$;
Assume, wlog, that $x_2$ is the left-most element  of $x_2, \ldots, x_k$ and is not equal to $x_1$. Then $d(f(1, x_2, \ldots, x_k), f(1, 1, y_3, \ldots, y_k))\le x_2-1$, with $y_i$ defined as above. By the inductive hypothesis $f(1, 1, y_3, \ldots, y_k)=1$, so that $d(f(1, x_2, \ldots, x_k), 1)\le x_2-1$, hence $f(1, x_2, \ldots, x_k)\le x_2$. 
We show that $f(1, x_2, \ldots, x_k)=1$ by induction on $x_2$.
If $x_2=2$ then $f(1, x_2, \ldots, x_k)\in \{1, 2\}$, and
we know that there is an  edge from $f(1, 2, x_3\ldots, x_k)$ to $f(2, 1, x_3-1, \ldots, x_k-1)$. Since there are no loops at $2$ and, by Case 1, $f(2, 1, x_3-1, \ldots, x_k-1)=2$, we must have $f(1, x_2, \ldots, x_k)=1$.

Assume now that the result holds whenever $x_2<z$. Then $f(1, z, x_3, \ldots, x_k)\le z$ and there is an arc from this vertex to $f(1, z-1, x_3-1, \ldots, x_k-1)$,  since $f(1, z-1, x_3-1, \ldots, x_k-1)=1$ by the inductive hypothesis, we must have $f(1, z, x_3, \ldots, x_k)\in \{1, 2\}$. 

Suppose, for a contradiction, that $f(1, z, x_3, \ldots, x_k)=2$. Since there is an arc from this vertex to $f(2, z-1, x_3-1, \ldots, x_k-1)$ we have that $f(2, z-1, x_3-1, \ldots, x_k-1)\in \{1, 3\}$. Now $d(f(2, z-1, x_3-1, \ldots, x_k-1), f(\lceil{z/2}\rceil+1, \lceil{z/2}\rceil, x_3-\lceil{z/2}\rceil,\ldots, x_k-\lceil{z/2}\rceil)\le \lceil{z/2}\rceil -1$. By Case 1, we know that $ f(\lceil{z/2}\rceil+1, \lceil{z/2}\rceil, x_3-\lceil{z/2}\rceil,\ldots, x_k-\lceil{z/2}\rceil)= \lceil{z/2}\rceil+1$. So, we cannot have $f(2, z-1, x_3-1, \ldots, x_k-1)=1$. It follows that $f(2, z-1, x_3-1, \ldots, x_k-1)=3$, and since there is an arc from this vertex to $f(1, z_2, x_3, \ldots, x_k)$ and, by the inductive hypothesis, $f(1, z_2, x_3, \ldots, x_k)=1$, we get a contradiction.

{\it Case 3:} $x_1$ is the left-most element of $x_1, \ldots, x_k$, but is not equal to $1$.

In this case we know that $d(f(1, x_2-x_1, \ldots, x_k-x_1), f(x_1, \ldots, x_k))\le x_1-1$, by Case 2 we know that $f(1, x_2-x_1, \ldots, x_k-x_1)=1$ it then follows that $f(x_1, \ldots, x_k)\le x_1$. Since we have already seen that $f(x_1, \ldots, x_k)\ge x_1$, because $x_1$ is the left-most element, it immediately follows that $f(x_1, \ldots, x_k)=x_1$.
This proves the claim.
\end{proof}

\subsubsection*{EGP Methodology via Galois correspondence}

\

The following is a restatement of the backward inclusion of the well-known Galois correspondence $\mathrm{Inv}(\mathrm{sPol}(\mathcal{B})) = \langle \mathcal{B} \rangle_{\mathrm{pH}}$ holding for finite structures $\mathcal{B}$ \cite{BBCJK}. This direction can be proved by induction on the term-complexity of $\phi \in \langle \mathcal{B} \rangle_{\mathrm{pH}}$.
\begin{lemma}
\label{lem:method}
Let $\mathcal{B}$ be a finite structure and suppose there is a $k$-ary surjective polymorphism of $\mathcal{B}$ that [pointwise] maps the tuples $(x^1_1,\ldots,x^r_1)$, \ldots, $(x^1_k,\ldots,x^r_k)$ to $(y^1,\ldots,y^r)$. Let $\phi$ be an $r$-ary relation from $\langle \mathcal{B} \rangle_{\mathrm{pH}}$. If $\phi$ holds on each of $(x^1_1,\ldots,x^r_1)$, \ldots, $(x^1_k,\ldots,x^r_k)$ in $\mathcal{B}$, then $\phi$ holds on $(y^1,\ldots,y^r)$ in $\mathcal{B}$.
\end{lemma}
Together with the definition of a generating set, it can be used to derive Corollary~\ref{cor:ppfuniondisguised}.



\section*{Material omitted from \S~\ref{sec:from-pgp-complexity}.}

\subsection*{Games, adversaries and reactive composition
  (\textsl{c.f.}\ref{sec:games-adversaries-reactivecomposition})}

\noindent \textbf{Theorem~\ref{thm:hubieReactiveComposition}.}
  Let $\phi$ be a pH-sentence with $m$ universal variables. Let $\mathscr{A}$ be an adversary and $\Omega_m$ a set of adversaries, both of length $m$.
  
  If $\mathcal{A}\models \phi_{\restrict\Omega_m}$ and  $\mathscr{A} \reactivelycomposable \Omega_m$ then
  $\mathscr{A} \models \phi$.

\begin{proof}
   We sketch the proof for the sake of completeness.
   Let $\Omega_m:=\{\mathscr{B}_1,\ldots,\mathscr{B}_k\}$ and $f$ and
   $g^i_j$ be as in the definition of reactive composition and
   witnessing that $\mathscr{A} \reactivelycomposable f(\mathscr{B}_1,\ldots,\mathscr{B}_k)$. 
   Assume also that $\mathcal{A}\models \phi_{\restrict\Omega_m}$.
   Given any sequence of play of the universal player according to the adversary $\mathscr{A}$, that is $v_1$ is played as $a_1 \in A_1$, $v_2$ is played as $a_2 \in A_2$, etc., we  "go backwards through $f$" via the maps $g^i_j$ to pinpoint \emph{incrementally} for each $j \in [k]$ a sequence of play $v_1=g^1_j(a_1)$, $v_2=g^2_j(a_1,a_2)$ etc, thus yielding eventually a tuple that belongs to adversary $\mathscr{B}_j$. After each block of universal variables, we lookup the winning strategy for the existential player against each adversary $\mathscr{B}_j$ and "going forward through $f$", that is applying $f$ to the choice of values for an existential variable against each adversary, we obtain a consistent choice for this variable against adversary $\mathscr{A}$ (this is because $f$ is a polymorphism and the quantifier-free part of the sentence $\phi$ is conjunctive positive). Going back and forth we obtain eventually an assignment to the existential variables that is consistent with the universal variables being played as $a_1,a_2,\ldots,a_m$.
\end{proof}
\begin{remark}
  In Chen's work on QCSP, constants are almost always allowed in the
  constraint language. This amounts with our definition to consider a
  relational structure $\mathcal{A}$ with all its elements named by
  constants.
  However, Chen does not necessarily explicitly add constants to the
  constraint language and instead moves rapidly to the algebraic
  setting and considers algebra. There he insists on additional
  technical conditions which preserves constants. For example  in the
  above theorem, he has the additional condition that $f$ is an
  \emph{idempotent} polymorphism.  Whenever we will use one of Chen's
  result, we will generalise it as above by considering arbitrary
  constraint languages and dropping technical conditions such as
  idempotency from the statement.
\end{remark}

\subsection*{The $\Pi_2$-case (\textsl{c.f.}\ref{sec:pi2})}
\noindent \textbf{Lemma~\ref{lemma:union} (principle of union).}
\textsl{Let $\Omega_m$ be a set of adversaries of length $m$ and $\phi$ a $\Pi_2$-sentence with $m$ universal variables.
Let $\mathscr{O}_{\cup\Omega}:=\bigcup_{\mathscr{O}\in \Omega} \mathscr{O}$
and $\Omega_{\text{tuples}}:=\{\{t\} \in \mathscr{O}_{\cup\Omega}\} = \bigcup_{\mathscr{O}\in \Omega}\{\{t\} \in \mathscr{O}\}$. 
We have the following equivalence.
$$\mathcal{A}\models \varphi_{\restrict\Omega_m} 
\quad \iff \quad
\mathcal{A}\models \varphi_{\restrict\mathscr{O}_{\cup\Omega}}
\quad \iff \quad
\mathcal{A}\models \varphi_{\restrict\Omega_{\text{tuples}}}$$
}

The forward implications
$$\mathcal{A}\models \varphi_{\restrict\Omega_m} 
\quad \implies \quad
\mathcal{A}\models \varphi_{\restrict\mathscr{O}_{\cup\Omega}}
\quad \implies \quad
\mathcal{A}\models \varphi_{\restrict\Omega_{\text{tuples}}}$$
of Lemma~\ref{lemma:union} hold clearly for arbitrary pH-sentences. 
The proof is trivial and is a direct consequence of the following
obvious fact.
\begin{fact}
  \label{fac:union}
  Let $\Omega_m$ be a set of adversaries of length $m$ and $\phi$ a $\Pi_2$-sentence with $m$ universal variables.
  $$\mathcal{A}\models \varphi_{\restrict\Omega_m}$$
  $$\quad \Updownarrow \quad$$
  $$\forall \mathscr{O} \in \Omega_m \forall t=(a_1,\ldots,a_m)\in  \mathscr{O} \ \mathcal{A}\models \varphi_{\restrict\{t\}}$$
\end{fact}

\begin{remark}[following Lemma~\ref{lemma:union}]
  For a sentence that is not $\Pi_2$, this does not necessarily hold. For example, consider $\forall x \forall y \exists z \forall w \ E(x,z) \wedge E(y,z) \wedge E(w,z)$ on the irreflexive $4$-clique $\mathcal{K}_4$. The sentence is not true, but for all individual tuples $(x_0,y_0,w_0)$, we have $\exists z \ E(x_0,z) \wedge E(y_0,z) \wedge E(w_0,z)$.
\end{remark}

\noindent \textbf{Proposition~\ref{prop:pi2composition}.}
  Let $\phi$ be a $\Pi_2$-pH-sentence with $m$ universal variables. Let $\mathscr{A}$ be an adversary and $\Omega_m$ a set of adversaries, both of length $m$.
  
  If $\mathcal{A}\models \phi_{\restrict\Omega_m}$ and  $\Omega_m$
  generates $\mathscr{A}$  then
  $\mathcal{A} \models \phi_{\restrict\mathscr{A}}$.

\begin{proof}
  The hypothesis that $\Omega_m$ generates $\mathscr{A}$
  can be rephrased as follows : for each tuple $t$ in $\mathscr{A}$,
  $\{t\}\reactivelycomposable f_t(t_1,t_2,\ldots,t_k)$, where
  $t_1,t_2,\ldots,t_k$ belong to $\Omega_{\text{tuples}}$. 
  To see this, it remains to note that the suitable $g^j_i$'s from the
  definition of composition are induced trivially as there is no
  choice: for every $j$ in $[k]$ and every $i$  
  in $[m]$ pick $g^j_i(a_1,a_2,\ldots,a_i)= t_{i,j}$ where $t_{i,j}$ is the ith
  element of $t_j$. 
  So by Theorem~\ref{thm:hubieReactiveComposition}, if
  $\mathcal{A}\models \phi_{\restrict\Omega_{\text{tuples}}}$ then
  $\mathcal{A}\models \phi_{\restrict\{t\}}$.
  As this holds for any tuple $t$ in $\mathscr{A}$, via the principle of union, it
  follows that $\mathcal{A} \models \phi_{\restrict\mathscr{A}}$.
\end{proof}

\noindent \textbf{Proposition~\ref{thm:characteringPi2}.}
  Let $\Omega_m$ be a set of adversaries of length $m$ that is not degenerate.
  The following are equivalent.
  \begin{romannum}
  \item for any $\Pi_2$-pH sentence $\psi$, $\mathcal{A}\models
    \psi_{\restrict\Omega_m}$ implies $\mathcal{A}\models \psi$.
  \item for any $\Pi_2$-pH sentence $\psi$, $\mathcal{A}\models \psi_{\restrict\mathscr{O}_{\cup\Omega}}$ implies $\mathcal{A}\models \psi$.
  \item for any $\Pi_2$-pH sentence $\psi$, $\mathcal{A}\models \psi_{\restrict\Omega_{\text{tuples}}}$ implies $\mathcal{A}\models \psi$.
  \item $\mathcal{A}\models
    \phi_{\mathscr{O}_{\cup\Omega},\mathcal{A}}$
  \item $\mathcal{A}\models \phi_{\Omega_{\text{tuples}},\mathcal{A}}$
  \item $\Omega_m$ generates $A^m$.
  \end{romannum}

\begin{proof}
  The first three items are equivalent by Lemma~\ref{lemma:union} (these implications have the same conclusion and equivalent premises).
  The fourth and fifth items are trivially equivalent since $\phi_{\mathscr{O}_{\cup\Omega},\mathcal{A}}$ and $\phi_{\Omega_{\text{tuples}},\mathcal{A}}$ are the same sentence.

  We show the implication from the third item to the fifth. By
  construction, $\phi_{\Omega_{\text{tuples}},\mathcal{A}}$ is $\Pi_2$
  and it suffices to show that there exists a winning strategy for
  $\exists$  against any adversary $\{t\}$ in
  $\Omega_{\text{tuples}}$. This is true by construction.
  Indeed, note that there exists a winning strategy for $\exists$ in
  the $(\mathcal{A},\phi_{\Omega_{\text{tuples}},\mathcal{A}})$-game
  against adversary  $\{t\}$ iff   there is a homomorphism from the
  $\sigma^{(m)}$-structure $\bigotimes_{t'\in \Omega_{\text{tuples}}} \mathfrak{A}_{\mu_{t'}}$
  to  the $\sigma^{(m)}$-structure $\mathfrak{A}_{\mu_t}$, where
  $\mu_t:[m]\to A$ is the map induced 
  naturally by $t$. The projection is such a homomorphism.
  
  The penultimate item implies the last one: instantiate the universal
  variables of $\phi_{\Omega_{\text{tuples}},\mathcal{A}}$ as given by
  the $m$-tuple $t$ and pick for $f_t$ the homomorphism from the product
  structure witnessing that $\exists$ has a winning strategy. 

  Finally, the last item implies the first one by
  Proposition~\ref{prop:pi2composition}. 
\end{proof}
\subsection*{The unbounded case (\textsl{c.f.}\ref{sec:unbounded} )}

\noindent \textbf{Lemma~\ref{lem:AdversariesWinnableOnCanonicalSentence}.}
  Let $\mathbf{\Omega}$ be a sequence of sets of adversaries that has
  the $m$-projectivity property for some $m\geq 1$ such that
  $\Omega_{n.m}$ is not degenerate. The following holds.

  \begin{romannum}
  \item 
    $\mathcal{A}\models \psi_{\restrict \Omega_{\mathbf{n.m}}}, \text{ where } \psi={\phi_{n,\Omega_\mathbf{m},\mathcal{A}}}$
  \item If for every $\Pi_2$-sentence $\psi$ with $m.n$ universal variables, it holds that $\mathcal{A}\models \psi_{\restrict \Omega_{\mathbf{m.n}}}$ implies $\mathcal{A}\models \psi$, then
    $\mathcal{A}\models \phi_{n,\Omega_\mathbf{m},\mathcal{A}}$.
  \end{romannum}

\begin{proof}
  The second statement is a direct consequence of the first one.
  The proof of the first statement generalises an argument used in the proof of Proposition~\ref{thm:characteringPi2}.
  Consider any adversary $\mathscr{O}$ in $\Omega_{n.m}$. For convenience, we name the positions of this adversary in a similar fashion to the universal variables of the sentence, namely by a pair $(i,j)$ in $[n]\times [m]$. By projectivity, there exists an adversary $\mathscr{O}'$ in $\Omega_{m}$ which dominates any adversary $\tilde{\mathscr{O}}$ in $\text{Proj}\mathscr{O}$ (obtained by projecting over an arbitrary choice of one position in each of the $m$ blocks of size $n$).
  In the product structure underlying the formula $\phi_{n,\Omega_m,\mathcal{A}}$, we consider the following structure:
  $$\bigotimes_{\mu \in A^{[n.m]}_{\upharpoonright \mathscr{O}'}} \mathfrak{A}_{\mathscr{O}',\mu}$$ 

  An instantiation of the universal variables of $\phi_{n,\Omega_m,\mathcal{A}}$ according to some tuple $t$ from the adversary $\mathscr{O}$ corresponds naturally to a map $\mu_t$ from $[n]\times [m]$ to $A$. Observe that our choice of $\mathscr{O}'$ ensures that this map $\mu_t$ is consistent with $\mathscr{O}'$.
  An instantiation of the universal variables by $\mu_t$ induces a $\sigma^{(n.m)}$-structure $\mathfrak{A}_{\mu_t}$ and a winning strategy for $\exists$ amounts to a homomorphism from the product $\sigma^{(n.m)}$-structure underlying the sentence to this $\mathfrak{A}_{\mu_t}$. Since the component $\mathfrak{A}_{\mathscr{O}',\mu_t}$ of this product structure is isomorphic to $\mathfrak{A}_{\mu_t}$, we may take for a homomorphism the corresponding projection. This shows that     $\mathcal{A}\models \psi_{\restrict \Omega_{\mathbf{n.m}}}$ where $\psi={\phi_{n,\Omega_\mathbf{m},\mathcal{A}}}$.
\end{proof}

\

\noindent \textbf{Theorem~\ref{MainResult:InAbstracto}.} (\textbf{In abstracto}.)
    Let $\mathbf{\Omega}$ be a projective sequence of adversaries,
    none of which are degenerate. 
    The following are equivalent.
  \begin{romannum}
  \item For every $m\geq 1$, For every pH-sentence $\psi$ with $m$
    universal variables, $\mathcal{A}\models \psi_{\restrict
      \Omega_{m}}$ implies $\mathcal{A}\models \psi$. 
  \item For every $m\geq 1$, for every $\Pi_2$-pH-sentence $\psi$ with $m$ universal
    variables, $\mathcal{A}\models \psi_{\restrict \Omega_{m}}$
    implies $\mathcal{A}\models \psi$.
  \item For every $m\geq 1$, $\mathcal{A}\models
    \phi_{n,\Omega_m,\mathcal{A}}$.
  \item For every $m\geq 1$, $\mathcal{A}\models
    \phi_{\mathscr{O}_{\cup\Omega},\mathcal{A}}$.
  \item For every $m\geq 1$, $A^m \reactivelycomposable \Omega_m$.
  \item For every $m\geq 1$, $\Omega_m$ generates $A^m$.
  \end{romannum}

\begin{proof}  
  Propositions~\ref{thm:characteringPi2} establishes the equivalence
  between  \ref{abstracto:logical:pi2}, \ref{abstracto:canonical:pi2} and
  \ref{abstracto:algebraic:pi2} for fixed values of $m$ (numbered
  there as \ref{pi2abstracto:logical:pi2}, \ref{pi2abstracto:canonical:pi2} and
  \ref{pi2abstracto:algebraic:pi2}, respectively).
  
  To lift these relatively trivial equivalences to the general case,
  the principle of our current proof no longer preserves the parameter $m$. The chain of
  implications of Theorem~\ref{theo:Pi2ToGeneral} translates here, once
  the parameter is universally quantified, to the chain of implications
  $$
  \ref{abstracto:logical:pi2}
  \implies 
  \ref{abstracto:canonical:general} 
  \implies 
  \ref{abstracto:algebraic:general}
  \implies 
  \ref{abstracto:logical:general}
  $$
  
  The fact that \ref{abstracto:logical:general} implies
  \ref{abstracto:logical:pi2} is trivial\footnotemark{}, which
  concludes the proof. 
  
  \footnotetext{We note in passing and for purely pedagogical reason that 
  the implication~\ref{abstracto:algebraic:general} to
  \ref{abstracto:algebraic:pi2} is also trivial, while the natural implication
  \ref{abstracto:canonical:general} to \ref{abstracto:canonical:pi2}
  will appear as an evidence to the reader once the definition of the canonical
  sentences is digested.}

\end{proof}

\

\noindent \textbf{Corollary~\ref{cor:EffectiveProjectivePGPImpliesQCSPInNP}.}
  Let $\mathcal{A}$ be a structure.  Let $\mathbf{\Omega}$ be a
  sequence of non degenerate adversaries that is effective, projective and
  polynomially bounded such that $\Omega_m$ generates $A^m$ for every
  $m\geq 1$.
  
  Let $\mathcal{A}'$ be the structure $\mathcal{A}$, possibly expanded
  with constants, at least one for each element that occurs in $\mathbf{\Omega}$.
  The problem $\QCSP(\mathcal{A})$ reduces in polynomial time to
  $\CSP(\mathcal{A}')$.
  In particular, if $\mathcal{A}$ has all constants, the problem $\QCSP_c(\mathcal{A})$ reduces in polynomial time to
  $\CSP_c(\mathcal{A})$.

 \begin{proof}
  To check whether a pH-sentence $\phi$ with $m$ universal variables
  holds in $\mathcal{A}$, by Theorem~\ref{MainResult:InAbstracto}, 
  we only need to check that $\mathcal{A}\models \phi_{\restrict
    \mathscr{B}}$ for every $\mathscr{B}$ in $\Omega_m$. 
  The reduction proceeds as in the proof
  of~\cite[Lemma~7.12]{AU-Chen-PGP}, which we outline
  here for completeness. 

  Pretend first that we reduce $\mathcal{A}\models \phi_{\restrict
    \mathscr{B}}$ to a collection of CSP instances, one for each tuple
  $t$ of $\mathscr{B}$, obtained by instantiation of the universal variables with the
  corresponding constants. If $x$ is an existential variable in
  $\phi$, let $x_t$ be the corresponding variable in the CSP instance
  corresponding to $t$. We will in fact enforce equality constraints
  via renaming of variables to
  ensure that we are constructing Skolem functions. For any two tuples
  $t$ and $t'$ in $\mathscr{B}$ that agree on their first $\ell$ coordinates, let
  $Y_\ell$ be the corresponding universal variables of $\phi$. For every
  existential variable $x$ such that $Y_x$ (the universally quantified
  variables of $\phi$ preceding $x$) is contained in $Y_\ell$, we identify
  $x_t$ with $x_{t'}$.
\end{proof}
\subsection*{Studies of Collapsibility (\textsl{c.f.}\ref{sec:collapse})}

\noindent\textbf{$p$-collapsibility for $p>0$}\\

\

\noindent \textbf{Application~\ref{app:FairmontHotelTrick}.}
  A partially reflexive path $\mathcal{A}$ (no constants are present)
  that is quasi-loop connected has the PGP. 

\begin{proof}
  Indeed, a partially reflexive path $\mathcal{A}$ that is quasi-loop
  connected has the same QCSP as a partially reflexive path that is
  loop-connected $\mathcal{B}$ \cite{DBLP:conf/cp/MadelaineM12} since
  for some $r_a>0$ there is a surjective homomorphism $g$ from
  $\mathcal{A}^{r_a}$ to $\mathcal{B}$ and for some $r_b>0$ there is a
  surjective homomorphism $h$ from $\mathcal{B}^{r_b}$ to
  $\mathcal{A}$ (see main result of~\cite{LICS2008}). We also know
  that $\mathcal{B}$ admits a majority polymorphism $m$
  \cite{QCSPforests} and is therefore $2$-collapsible from any
  singleton source (see Table~\ref{tab:Polymorphismscollapsibility})
  and that Theorem~\ref{MainResult:InConcreto:Collapsibility} holds
  for $\mathcal{B}$.  Pick some arbitrary element $a$ in $\mathcal{A}$
  such that there is some $b$ in $\mathcal{B}$ satisfying
  $g(a,a,\ldots,a)=b$. Use $b$ as a source for $\mathcal{B}$.

  We proceed to lift~\ref{concreto:algebraiccollapsibility:pi2} of
  Corollary~\ref{MainResult:InConcreto:Collapsibility} from
  structure $\mathcal{B}$ to $\mathcal{A}$, which we recall here for
  $\mathcal{B}$ : for every $m$, for every tuple $t$ in $B^m$, there
  is a polymorphism $f_t$ of $\mathcal{B}$ of arity $k$ and tuples
  $t_1,t_2,\ldots,t_k$ in $\Upsilon_{m,2,b}$ such that
  $f_t(t_1,t_2,\ldots,t_k)=t$.

  Let $g^{k}$ denote the surjective homomorphism from
  $(\mathcal{A}^{r_a})^k$ to $\mathcal{B}^k$ that applies $g$
  blockwise.  Going back from $t_i$ through $g$, we can find $r_a$
  tuples $t_{i,1},t_{i,2},\ldots,t_{i,r_a}$ all in $\Upsilon_{m,2,a}$
  (adversaries based on the domain of $\mathcal{A}$) such that
  $g(t_{i,1},t_{i,2},\ldots,t_{i,r_a})=t_i$.  Thus, we can generate
  any $\widetilde{t}$ in $\mathcal{B}$ via $f_{\widetilde{t}}\circ(g^k)$ from
  tuples of $\Upsilon_{m,2,a}$.

  Let $\hat{t}$ be now some tuple of $\mathcal{A}$. By surjectivity of
  $h$, let 
  $\widetilde{t_1},\widetilde{t_2},\ldots,\widetilde{t_{r_b}}$
  be tuples of $\mathcal{B}$ such that
  $h(\widetilde{t_1},\widetilde{t_2},\ldots,\widetilde{t_{r_b}})=\hat{t}$. 
  The polymorphism of $\mathcal{A}$ 
  $(f_{\widetilde{t_1}}\circ(g^k),f_{\widetilde{t_2}}\circ(g^k),\ldots,f_{\widetilde{t_{r_b}}}\circ(g^k))$
  shows that $\Upsilon_{m,2,a}$ generates $\hat{t}$.
  This shows that $\mathcal{A}$ is also 2-collapsible from a
  singleton source. 
\end{proof}

\

\noindent \textbf{Lemma~\ref{lem:ChensLemma}.} (Chen's lemma.)
  Let $\mathcal{A}$ be a structure with a constant $x$.
  if there is a $k$-ary polymorphism of $\mathcal{A}$ such that $f$ is surjective when restricted at any position to $\{x\}$, then $\mathcal{A}$ is $k-1$-collapsible from source $\{x\}$  (\mbox{i.e.} $\mathcal{A}$ as a $k$-ary Hubie polymorphism).
\begin{proof}

  We sketch the proof for pedagogical reasons.
  Via Corollary~\ref{MainResult:InConcreto:Collapsibility}, it suffices to show
  that for any $m$, $A^m$ is generated by $\Upsilon_{m,k-1,x}$
  (instead of the notion of reactive composition). 

  Consider adversaries of length $m=k$ for now, that is from $\Upsilon_{k,k-1,x}$. 
  If we apply $f$ to these $k$ adversaries, we generate the full adversary $A^k$. With a picture (adversaries are drawn as columns):
  $$f
  \begin{pmatrix}
    \{x\}  &A      &A      &\ldots & A      \\
    A      &\{x\}  &A      &\ldots & A      \\
    \vdots &       &\ddots &       & \vdots \\
    A      &\ldots & A     & \{x\} & A      \\
    A      &\ldots & A     & A     & \{x\} \\
  \end{pmatrix}
  = 
  \begin{pmatrix}
   A \\
   A \\
   \vdots \\
   A \\
   A \\
  \end{pmatrix}
  =
  A^k
  $$
  Expanding these adversaries uniformly with singletons $\{x\}$ to the
  full length $m$, we may produce an adversary from $\Upsilon_{m,k,x}$. 
  With a picture for \textsl{e.g.} trailing singletons:
  $$f
  \begin{pmatrix}
    \{x\}  &A      &A      &\ldots & A      \\
    A      &\{x\}  &A      &\ldots & A      \\
    \vdots &       &\ddots &       & \vdots \\
    A      &\ldots & A     & \{x\} & A      \\
    A      &\ldots & A     & A     & \{x\} \\
    \{x\}  &\{x\}  &\{x\}  &\ldots & \{x\}  \\
    \vdots &\vdots &\vdots &\vdots &\vdots  \\
    \{x\}  &\{x\}  &\{x\}  &\ldots & \{x\}  \\
  \end{pmatrix}
  = 
  \begin{pmatrix}
    A \\
    A \\
    \vdots \\
    A \\
    A \\
    \{x\} \\
    \vdots \\
    \{x\} \\
  \end{pmatrix}
  $$
  Shifting the first additional row of singletons in the top block, we
  will obtain the family of adversaries from $\Upsilon_{m,k,x}$ with
  a single singleton in the first $k+1$ positions. It should be now
  clear that we may iterate this process to derive $A^m$ eventually
  via some term $f'$ which is a superposition of $f$ and projections
  and is therefore also a polymorphism of $\mathcal{A}$.
\end{proof}

\begin{remark}
  An extended analysis of our proof should convince the careful reader that
  we may in the same fashion prove retroactive composition (the polymorphism's
  action is determined for a row independently of the others). Thus, appealing to
  the previous section is not essential, though it does allow for a
  simpler argument. 
\end{remark}

\noindent \textbf{Proposition~\ref{prop:singletoncollapse:WeakCharacterisation}.}
  Let $x$ be a constant in $\mathcal{A}$.
  The following are equivalent:
  \begin{romannum}
  \item $\mathcal{A}$ is collapsible from $\{x\}$. 
  \item $\mathcal{A}$ has a Hubie polymorphism.
  \end{romannum}

\begin{proof}
  Lemma~\ref{lem:ChensLemma} shows that \ref{item:HubiePolx} implies
  collapsibility. We prove the converse.
  
  Assume $p$-collapsibility. By
  Fact~\ref{fact:collapsible:adversary:is:projective}, we may apply
  Theorem~\ref{MainResult:InAbstracto}. 
  For $m = p+1$, item \ref{abstracto:algebraic:general} of this theorem states that
  there is a polymorphism $f$ witnessing that $A^{p+1} \reactivelycomposable
  \Upsilon_{p+1,p,x}$ (diagrammatically, we may draw a similar picture
  to the one we drew at the beginning of the previous proof). Clearly,
  $f$ satisfies~\ref{item:HubiePolx}.
\end{proof}

\

\noindent \textbf{Theorem~\ref{theo:singletonSource:StrongCharacterisation}.} (\textbf{$p$-Collapsibility from a singleton source}).
  Let $x$ be a constant in $\mathcal{A}$.
  The following are equivalent:
  \begin{longromannum}
  \item $\mathcal{A}$ is $p$-collapsible from $\{x\}$.
  \item For every $m\geq 1$, the full adversary $A^m$ is reactively composable from $\Upsilon_{m,p,x}$.
  \item $\mathcal{A}$ is $\Pi_2$-$p$-collapsible from $\{x\}$. 
  \item For every $m\geq 1$, $\Upsilon_{m,p,x}$ generates $A^m$.
  \item $\mathcal{A}$ models $\phi_{n,\Upsilon_{p+1,p,x},\mathcal{A}}$
    (which implies that $\mathcal{A}$ admits a particularly well
    behaved Hubie polymorphism with source $x$ of arity $(p+1)n^p$).
  \end{longromannum}

\begin{proof}
  Equivalence of the first four points appears in
   Corollary~\ref{MainResult:InConcreto:Collapsibility}, 
  as does the
  equivalence with the statement :
  For every $m \geq 1$, $\mathcal{A}$ models $\phi_{n,\Upsilon_{m,p,x},\mathcal{A}}$.
  So they imply trivially the last point by selecting $m = p+1$.
  
  We show that the last point implies the penultimate one. The proof principle is similar to that of Chen's Lemma.
  As we have argued similarly before, the last point implies the existence of a polymorphism $f$.
  This polymorphism enjoys the following property (each column represents in fact $n^p$ coordinates of $A$):
  $$f
  \begin{pmatrix}
    \begin{array}[h]{c|c|c|c|c}
    \{x\}   &A      &A      &\ldots & A      \\
    A      &\{x\}  &A      &\ldots & A      \\
    \vdots &       &\ddots &       & \vdots \\
    A      &\ldots & A     & \{x\} & A      \\
    A      &\ldots & A     & A     & \{x\} \\
  \end{array}
  \end{pmatrix}
  = 
  \begin{pmatrix}
   A \\
   A \\
   \vdots \\
   A \\
   A \\
  \end{pmatrix}
  =
  A^{p+1}
  $$
  So arguing as in the proof of Chen's Lemma, we may conclude similarly that 
  for all $m$, the full adversary $A^m$ is composable from $\Upsilon_{m,p,x}$.
\end{proof}

\noindent\textbf{$p$-collapsibility for $p>0$ from a conservative source}\\
We expand on Remark~\ref{rem:collapsibilityOnConservativeSourceSet}.

\begin{theorem}[\textbf{$p$-Collapsibility from a conservative source}]
  Let $B$ be a subset of the domain of a structure $\mathcal{A}$. Assume further that $\mathcal{A}$ is $B$-conservative.

  The following are equivalent:
  \begin{longromannum}
  \item $\mathcal{A}$ is $p$-collapsible from $B$.
  \item $\mathcal{A}$ models $\phi_{n,\Upsilon_{p+1,p,B},\mathcal{A}}$
    (which implies that $\mathcal{A}$ admits a polymorphism $f$ of arity
    $|B|(p+1)n^p$ that remains surjective when a position $i$ is fixed to
    a suitable source element $b_i$ in $B$, and that this polymorphism
    witnesses that $A^m$ is reactively composable from $\Upsilon_{m,p,B}$).
  \end{longromannum}
\end{theorem}
\begin{proof}
  Just like the case of singleton source, almost all the proof follows directly
  from  Corollary~\ref{MainResult:InConcreto:Collapsibility} and similarly we shall only need
  to prove that the last point  implies the penultimate one via a
  bootstrapping argument. 

  As we have argued similarly before, the last point implies the
  existence of a polymorphism $f$.
  Let $x_1,x_2,\ldots,x_b$ enumerate the elements of the source $B$.
  This polymorphism enjoys the following property (each column represents in fact $n^p$ coordinates of $A$):
  \resizebox{\columnwidth}{!}{$\displaystyle\tiny f
  \begin{pmatrix}
    \begin{array}[h]{c|c|c|c|c|c|c|c|c|c|c|c|c}
    \{x_1\}  &\{x_2\}  & \ldots & \{x_b\} &
    A        &A       &\ldots  & A       &
                       \ldots            & 
    A        &A       &\ldots  & A      \\
    A        &A       &\ldots  & A        &
    \{x_1\}  &\{x_2\}  & \ldots & \{x_b\} &
                       \ldots             & 
    A        &A       &\ldots  & A      \\
    A        &A       &\ldots  & A        &
    A        &A       &\ldots  & A        &
                                    &
    A        &A       &\ldots  & A        \\
    \vdots   &\vdots  &\vdots  & \vdots   &
    \vdots   &\vdots  &\vdots  & \vdots   &
                       \ddots             &
    \vdots   &\vdots  &\vdots  & \vdots   \\
    A        &A       &\ldots  & A        &
    A        &A       &\ldots  & A        &
                       \ldots             &
    \{x_1\}  &\{x_2\}  & \ldots & \{x_b\} \\
  \end{array}
  \end{pmatrix}
  = 
  \begin{pmatrix}
   A \\
   A \\
   A \\
   \vdots \\
   A \\
  \end{pmatrix}
  =
  A^{p}
  $}
  By conservativity, $f(x_1,x_2 \ldots x_b, \ldots,x_1,x_2 \ldots
  x_b)\in B$ and we may assume w.l.o.g. that it is in fact equal to
  $x_1$. 
  So adding this line at the bottom of the above we may obtain the
  tuple $(A^{p},x_1)$ and (similarly for the other permutations of
  $x_1$ within $A^p$) :
  \resizebox{\columnwidth}{!}{$\displaystyle\tiny f
  \begin{pmatrix}
    \begin{array}[h]{c|c|c|c|c|c|c|c|c|c|c|c|c}
    \{x_1\}  &\{x_2\}  & \ldots & \{x_b\} &
    A        &A       &\ldots  & A       &
                       \ldots            & 
    A        &A       &\ldots  & A      \\
    A        &A       &\ldots  & A        &
    \{x_1\}  &\{x_2\}  & \ldots & \{x_b\} &
                       \ldots             & 
    A        &A       &\ldots  & A      \\
    A        &A       &\ldots  & A        &
    A        &A       &\ldots  & A        &
                                    &
    A        &A       &\ldots  & A        \\
    \vdots   &\vdots  &\vdots  & \vdots   &
    \vdots   &\vdots  &\vdots  & \vdots   &
                       \ddots             &
    \vdots   &\vdots  &\vdots  & \vdots   \\
    A        &A       &\ldots  & A        &
    A        &A       &\ldots  & A        &
                       \ldots             &
    \{x_1\}  &\{x_2\}  & \ldots & \{x_b\} \\
    \{x_1\}  &\{x_2\}  & \ldots & \{x_b\} &
    \{x_1\}  &\{x_2\}  & \ldots & \{x_b\} &
                         \ldots           &
    \{x_1\}  &\{x_2\}  & \ldots & \{x_b\} \\
  \end{array}
  \end{pmatrix}
  = 
  \begin{pmatrix}
   A \\
   A \\
   A \\
   \vdots \\
   A \\
   \{x_1\} \\
  \end{pmatrix}
  $}
  Now, we may recopy the above picture replacing in all columns with $x_1$ at least one
  of the two occurrences of $x_1$ by $A$ (we have all permutation of
  tuples of the form $(A^p,x_1)$). In particular, we may chose
  for the last line, the value $x_2$. Assuming w.l.o.g. that the image
  of the last line is $x_2$ (by $B$-conservativity). We obtain that :
  \resizebox{\columnwidth}{!}{$\displaystyle\tiny f
  \begin{pmatrix}
    \begin{array}[h]{c|c|c|c|c|c|c|c|c|c|c|c|c}
    \{x_1\}  &\{x_2\}  & \ldots & \{x_b\} &
    A        &A       &\ldots  & A       &
                       \ldots            & 
    A        &A       &\ldots  & A      \\
    A        &A       &\ldots  & A        &
    \{x_1\}  &\{x_2\}  & \ldots & \{x_b\} &
                       \ldots             & 
    A        &A       &\ldots  & A      \\
    A        &A       &\ldots  & A        &
    A        &A       &\ldots  & A        &
                                    &
    A        &A       &\ldots  & A        \\
    \vdots   &\vdots  &\vdots  & \vdots   &
    \vdots   &\vdots  &\vdots  & \vdots   &
                       \ddots             &
    \vdots   &\vdots  &\vdots  & \vdots   \\
    A        &A       &\ldots  & A        &
    A        &A       &\ldots  & A        &
                       \ldots             &
    \{x_1\}  &\{x_2\}  & \ldots & \{x_b\} \\
    \{x_2\}  &\{x_2\}  & \ldots & \{x_b\} &
    \{x_2\}  &\{x_2\}  & \ldots & \{x_b\} &
                         \ldots           &
    \{x_2\}  &\{x_2\}  & \ldots & \{x_b\} \\
  \end{array}
  \end{pmatrix}
  = 
  \begin{pmatrix}
   A \\
   A \\
   A \\
   \vdots \\
   A \\
   \{x_2\} \\
  \end{pmatrix}$}
  Iterating this trick, replacing this time the last occurrence of
  $x_1$ and $x_2$ (from our original picture) by $x_3$, we will obtain
  a value in $B$ that differs from $x_1$ and $x_2$, say $x_3$
  w.l.o.g. Eventually, we show that $A^{p+2}$ may be generated from 
  $\Upsilon_{o+2,p,B}$.
  Iterating this bootstrapping technique for higher arity, we show that
  for any $m$, the full adversary $A^m$ may be generated from $\Upsilon_{m,p,B}$.
\end{proof}

\begin{corollary}
  Given $p \geq 1$, a structure $\mathcal{A}$ that is $B$-conservative, we may decide
  whether $\mathcal{A}$ is $p$-collapsible from source $B$.
\end{corollary}

\noindent\textbf{$0$-collapsibility} (\textsl{proofs  were omitted fully from
  paper})

\

\noindent \textbf{Theorem~\ref{theorem:0collapsibility}.}
  Let $\mathcal{B}$ be a finite structure.
  The following are equivalent.
  \begin{romannum}
  \item $\mathcal{B}$ is $0$-collapsible from source $\{x\}$ for some
    $x$ in $B$.
  \item $\mathcal{B}$ admits a simple $A$-she.
  \item $\mathcal{B}$ is $0$-collapsible for sentences of positive
    equality free first-order logic from source $\{x\}$ for some
    $x$ in $B$. 
  \end{romannum}

\begin{proof}
  The last two points are equivalent~\cite[Theorem
  8]{DBLP:conf/lics/MadelaineM09} (this result is stated with
  $A$-she rather than simple $A$-she but clearly, $\mathcal{A}$ has an A-she iff it has a simple A-she).
  The implication \ref{0collapsibility:she} to
  \ref{0collapsibility:singleton:pH} 
  follows trivially. 
  
  We prove the implication \ref{0collapsibility:singleton:pH} to
  \ref{0collapsibility:she} by contraposition.
  Assume that $A=[n]=\{1,\ldots,n\}$ and suppose that $\mathcal{A}$ has no simple A-she. We will prove that
  $\mathcal{A}$ does not admit universal relativisation to $x$
  for pH-sentences. We assume also \textsl{w.l.o.g.} that $x=1$.
  Let $\Xi$ be the set of simple A-shops $\xi$ \mbox{s.t.}
  $\xi(1)=[n]$. Since each $\xi$ is not a she of $\mathcal{A}$, we
  have a quantifier-free formula with $2n-1$ variables $R_\xi$ that
  consists of a single positive atom (not all variables need appear
  explicitly in this atom) such that $\mathcal{A} \models R_\xi(1,\ldots,1,2,\ldots,n)$\footnote{There are $n$ ones.}, but $\mathcal{A} \notmodels R_\xi(\xi^1,\ldots,\xi^n,\xi(2),\ldots,\xi(n))$ for some $\xi^1,\ldots,\xi^n \in [n]=\xi(1)$.

  \newcommand{\Eta}{\mathrm{E}}
  This means that for each $\eta:\{2,\ldots,n\} \rightarrow
  [n]$ there is some $2n-1$-ary ``atom'' $R_\eta$ such that
  $\mathcal{A} \models R_\eta(1,\ldots,1,1,2,\ldots,n)$\footnote{There
    are $n$ ones.}, but $\mathcal{A} \notmodels
  R_\eta(\xi^1,\ldots,\xi^n,\eta(2),\ldots,\eta(n))$ for some
  $\xi^1,\ldots,\xi^n \in [n]$. Let $\Eta=[n]^{[n-1]}$ denotes the set
  of $\eta$s.

  Suppose we had universal relativisation to $1$. Then we know that 
\[ \mathcal{A} \models \bigwedge_{\eta \in \Eta} R_\eta(1,\ldots,1,1,2,\ldots,n), \] that is, 
\[ \mathcal{A} \models \exists y_1,\ldots,y_n \bigwedge_{\eta \in \Eta} R_\eta(1,\ldots,1,y_1,y_2,\ldots,y_n).\]
According to relativisation this means also that 
\[ \mathcal{A} \models \exists y_1,\ldots,y_n \forall x_1, \ldots,x_n \bigwedge_{\eta \in \Eta} R_\eta(x_1,\ldots,x_n,y_1,y_2,\ldots,y_n).\]
But we know 
\[ \mathcal{A} \models \forall y_1,\ldots,y_n \exists x_1, \ldots,x_n \bigvee_{\eta \in \Eta} \neg R_\eta(x_1,\ldots,x_n,y_1,y_2,\ldots,y_n),\]
since the $\eta$s range over all maps $[n]$ to $[n]$. Contradiction.
\end{proof}

\noindent \textbf{Theorem~\ref{theorem:0collapsibility:rainbowsource}.}
  Let $\mathcal{B}$ be a structure. The following are equivalent.
  \begin{romannum}
  \item $\mathcal{B}$ is $0$-collapsible from source $C$
  \item $\mathcal{B}^{|C|}$ is $0$-collapsible from some (any) singleton source $x$
    which is a (rainbow) $|C|$-tuple containing all elements of $C$.
  \end{romannum}

\begin{proof}
  Let $B=\{1,2,\ldots,b\}$.
  \begin{itemize}
  \item (downwards). 
    Let $x$ be $|B|$-tuple containing all elements of $B$, wlog $x=(1,2,\ldots,b)$.
    Let $\varphi$ be a pH sentence.
    Assume that $\mathcal{A}^{|B|}\models \varphi_{\restrict
      (x,x,\ldots,x)}$.
    Equivalently, for any $i$ in $B$, $\mathcal{A}\models \varphi_{\restrict
      (i,i,\ldots,i)}$. Thus, $0$-collapsibility from source $B$
    implies that $\mathcal{A}\models \varphi$. Since $A$ and its power
    satisfy the same
    pH-sentences\cite{LICS2008,DBLP:journals/corr/ChenMM13}
    we may conclude that $\mathcal{A}^{|B|}\models \varphi$.
  \item (upwards). Assume that for any $i$ in $B$, $\mathcal{A}\models \varphi_{\restrict
      (i,i,\ldots,i)}$. Equivalently, $\mathcal{A}^{|B|}\models \varphi_{\restrict
      (x,x,\ldots,x)}$ where $x$ is any $|B|$-tuple containing all
    elements of $B$. By assumption, $\mathcal{A}^{|B|}\models \varphi$
    and we may conclude that $\mathcal{A}\models \varphi$.
  \end{itemize}
\end{proof}



\end{document}